\documentclass[12pt,a4paper]{article}

\usepackage[british]{babel}
\usepackage{titlesec}
\usepackage[a4paper,top=4cm,bottom=2cm,left=2.5cm,right=2.5cm,marginparwidth=1.75cm]{geometry}
\usepackage{graphicx}
\usepackage{bbm}
\usepackage{xcolor}
\usepackage{amsmath}
\usepackage{amsthm}
\usepackage{setspace}
\usepackage{amsfonts}
\usepackage{csquotes}
\usepackage{mathrsfs}
\usepackage{dsfont}
\usepackage{appendix}
\usepackage{comment}
\usepackage{footmisc}
\usepackage{mathtools}
\usepackage{float}
\usepackage{pgfplots}
\pgfplotsset{compat=1.18}
\usetikzlibrary{arrows.meta}
\usepackage{tikz}
\usetikzlibrary{arrows.meta,calc}
\usepgfplotslibrary{fillbetween}
\definecolor{cBlue}{RGB}{0,83,159}
\definecolor{cOrange}{RGB}{213,94,0}
\definecolor{cGray}{RGB}{90,90,90}
\usepackage{placeins} 

\theoremstyle{plain}
\newtheorem{theo}{Theorem}
\newtheorem{prop}{Proposition}
\newtheorem{coro}{Corollary}
\newtheorem{claim}{Claim}

\newtheorem{lemma}{Lemma}

\theoremstyle{definition}

\newtheorem{eg}{Example}

\newtheorem{rmk}{Remark}

\newcommand{\citeposs}[1]{\citeauthor{#1}'s (\citeyear{#1})}

\usepackage[authoryear]{natbib}
\setcitestyle{round,aysep={,}}
\usepackage[colorlinks=true,linkcolor=blue,citecolor=blue,urlcolor=blue]{hyperref} 
\title{An Explicit Solution for the Two-Bidder All-Pay Auction with Arbitrary Type Distributions\thanks{I am deeply indebted to Charles Z. Zheng for invaluable guidance and support throughout the project. I also thank Roy Allen, Elizabeth Caucutt, Christopher P. Chambers, Sergio Ocampo Diaz, Maria Goltsman, Ren{\'e} Kirkegaard, Jingfeng Lu, Gregory Pavlov, Ron Siegel and Peter A. Streufert for helpful comments. I use ChatGPT and Claude for preparation of the manuscript. All remaining errors are my own.}}
\author{Yijia Ding\thanks{yding463@uwo.ca. Department of Economics, University of Western Ontario.}}
\date{Sept 12, 2026}

\begin{document}

\maketitle

\begin{abstract}
This paper provides an explicit characterization of the (unique) equilibrium of the two-bidder all-pay auction with asymmetrically distributed, independent private values. The characterization is achieved through recasting the problem into its equivalence in the quantile space and thereby obtaining equilibrium bid functions that involve only integrals and inverses of the primitives. No regularity assumption is required: the value distributions can be continuous, discrete, singular continuous, or any mixture of these. The explicit formula makes comparative statics transparent. Raising one bidder’s value distribution in first-order stochastic dominance raises that bidder’s equilibrium bid distribution in the same order. Perturbing a symmetric-bidder model has different revenue effects depending on the direction of the perturbation: whereas slightly weakening one of the two bidders reduces expected revenue, slightly strengthening one of the two has an ambiguous effect. The explicit formula also turns some involved arguments in the literature into direct corollaries.
\end{abstract}

\section{Introduction}

This paper studies the two-bidder all-pay auction with asymmetric, independently distributed private values. This has been a major model in the study of conflicts, contests, patent races and political lobbying, but its equilibrium has been difficult to use in full generality. The classical analysis of \citet{amann1996asymmetric} studies the case in which value distributions have strictly positive densities on their supports, while \citet{siegel2014asymmetric} studies finite type distributions at the other extreme. These approaches characterize the equilibrium, but they leave two limitations. First, in neither case is the equilibrium written as an explicit function of the primitives. Second, the analysis relies on either smooth continuous distributions or finite discrete distributions. This paper removes both restrictions in the two-bidder independent private value environment.

The main result in this paper gives explicit, closed-form formulas for the equilibrium bid functions in terms of only the primitives (Theorem~\ref{main theo}). Because the equilibrium is exhibited explicitly, its existence and uniqueness follow from the construction itself, without the fixed-point arguments common in this literature \citep[e.g.,][]{prokopovych2023monotone}. No continuity or discreteness assumption is imposed on the value distributions: each distribution may be continuous, discrete, singular continuous, or any mixture of these cases. This generality matters for applications such as multi-stage all-pay auctions, where off-path posteriors are unrestricted by Bayes's rule, and information design, where imposing continuity or discreteness on posterior distributions would restrict the designer's feasible information structures.

The organizing idea is to work in quantile space. Instead of describing a strategy as a map from values to bid distributions, we rank each bidder's value by its quantile and describe bidding as a function of this quantile. This is especially useful when the value distribution has atoms. An atom corresponds to a whole interval of quantiles for the same value of a bidder, and this interval can absorb the randomization that the value type uses in a mixed-strategy equilibrium. The quantile representation thus serves as a purification device in the spirit of \citet{milgrom1985distributional}: Lemma~\ref{increasing bid function} shows that an equilibrium that is possibly in mixed strategies in the original value representation can then be represented, without changing the induced joint distribution of values and bids, by deterministic monotone bid functions of quantiles.

The quantile method has been applied to auction theory \citep{bulow1989simple,alaei2013simple}, and \citet{lu2017monotone} have used a closely related transformation in a two-bidder all-pay auction model.\footnote{\citet{lu2017monotone} transform each bidder's signal into a uniform random variable on $[0,1]$, with the bidder's value becoming a function of the transformed signals. Due to their interdependent value assumption, an explicit solution of the equilibrium is not available in their model.} It is natural to expect that the quantile method could generalize the classical analysis of \citet{amann1996asymmetric}, which characterizes the equilibrium through an ordinary differential equation (ODE) under smooth value distributions. Carrying out the generalization, however, is not immediate, and the literature has not shown how it could be done or how far it could go. This paper lays out the path formally and shows that, with the simplicity the quantile method provides, the generalization requires no restriction other than the assumption of independent private values.

The main difficulty is that, with arbitrary value distributions, the first-order condition for the equilibrium bid functions need not hold at every quantile. The paper resolves it through convex analysis, exploiting only the independent private value assumption.\footnote{The same difficulty is confronted by \citet{zheng2019necessary}, who resolves it to recover the equilibrium bid distribution from the first-order condition, with types kept in the value space. The approach here instead recasts both values and bids in quantile terms, so the equilibrium is a pair of monotone bid functions on the same unit interval of quantiles. This unified representation is what makes the explicit formula in Theorem~\ref{main theo} attainable.}

The explicit nature of the solution is the first contribution of the paper. Previous analyses of asymmetric auctions often characterize the equilibrium through endogenous variables rather than primitives: it is constructed through a finite-step algorithm \citep{siegel2014asymmetric,rentschler2016two},\footnote{\citet{siegel2014asymmetric}, \citet{rentschler2016two} and \citet{lu2017monotone} allow interdependence and correlation between the two bidders' values, while this paper does not.} or pinned down as the solution of an ODE, as above. These characterizations are self-referential \citep[e.g.,][]{maskin2000asymmetric,lebrun1999first} in the sense that the solution is defined by equations or procedures involving the solution itself, rather than by a direct formula in terms of primitives. Except in special cases, such as symmetric bidders or the recent finite-type analysis of \citet{chen2025regularization}, they do not yield explicit bid functions from the primitive distributions. This makes applications and comparative statics less transparent. By contrast, Theorem~\ref{main theo} expresses the equilibrium directly from primitives. Not only does this make the equilibrium easier to compute, but it also makes comparative statics that would otherwise be opaque transparent.

The generality is the second contribution of the paper: the characterization requires no regularity assumption on the value distributions. It matters because atoms and continuous parts can arise naturally in applications. In multi-stage all-pay auction models \citep[e.g.,][]{zheng2019necessary,lu2026peacepayments}, off-path histories may leave posterior beliefs unrestricted by Bayes's rule, so the continuation all-pay auction may have to be analyzed under posterior distributions that are discrete, continuous, or mixtures of the two. In models of information disclosure \citep[e.g.,][]{morath2008private,lu2018ranking} or information design \citep[e.g.,][]{chen2017persuasion}, the posterior distributions are determined by the designer's choice of information structure, so an equilibrium characterization that covers only continuous or only discrete posteriors leaves a substantial part of the designer's feasible set outside the analysis. The present formula allows these cases to be treated within a single framework.\footnote{Some related papers impose symmetry, binary types, or other special structure; see, for example, \citet{fu2014disclosure} and \citet{lu2018ranking}. The results here may be useful for extending such analyses to asymmetric environments with more general type distributions.}

Theorem~\ref{main theo} is closely related to \citet{chen2025regularization}, who recently obtain closed-form equilibrium expressions when each of the two bidders has finitely many types. Their analysis centers on the optimal design of all-pay auctions through favoritism and information disclosure, and the equilibrium is characterized through a regularization procedure. The procedure aligns the two bidders' types into common divisions through finitely many recursive steps and expresses the equilibrium as finite sums over these divisions, using the finite-type analysis of \citet{siegel2014asymmetric}. Each of these elements relies on the finiteness of the type space, so the procedure is confined to the purely discrete case, which is the relevant environment for their design analysis.\footnote{One may conjecture that the general result could be obtained from the discrete case through an approximation argument. Making the argument rigorous, however, would require establishing the convergence of the discrete equilibria and identifying the limit, which would involve more technical issues than the direct analysis in this paper. Proving the uniqueness of the equilibrium through a limit argument raises a further difficulty, since different approximating sequences need not select the same limit.} My paper aims at the equilibrium characterization itself, for arbitrary value distributions. In quantile space, the counterpart of their division alignment is a matching equation solved in closed form by integration. Applied to purely discrete type distributions, the method recovers \citeposs{chen2025regularization} procedure: their divisions correspond to subintervals of the quantile space, and their alignment condition is the discrete form of the matching equation.

The remaining results exploit the explicit formula. Section~\ref{sec:Bid distribution comparison} shows that if one bidder's value distribution is strengthened in the sense of first-order stochastic dominance (FOSD) while the other bidder's distribution is held fixed, then the strengthened bidder's equilibrium bid distribution also increases in the FOSD order (Theorem~\ref{theo:bid fosd}). Consequently, every type of the bidder whose distribution is held fixed is weakly worse off when facing the strengthened opponent (Corollary~\ref{coro:opponent payoff}). The economic content of both statements is established by \citet{kirkegaard2008comparative} for value distributions that are continuously differentiable with strictly positive densities. The contribution of the versions here lies in their generality: they dispense with all regularity assumptions on the value distributions. A companion result (Proposition~\ref{prop:FOSD}) shows that within a single auction, if one bidder's value distribution dominates the other's in the FOSD sense, the corresponding equilibrium bid distributions also inherit the same relation.

Taking expectations of the bid functions yields a closed-form expression for expected revenue, which we use to ask how revenue responds when one of two symmetric bidders becomes slightly stronger or weaker (Proposition~\ref{prop:local revenue effects} in Section~\ref{sec:revenue comparison}). Such a perturbation creates two effects on expected revenue. One is a level effect in shifting the overall level of values in the auction. The other is an asymmetry effect in making the two bidders asymmetric in their strength. The level effect raises expected revenue when the bidder is strengthened and lowers it when the bidder is weakened, whereas the asymmetry effect lowers expected revenue in either direction. The response is therefore one-sided: locally weakening one bidder always lowers expected revenue, whereas locally strengthening one bidder may either raise or lower it. This generalizes Proposition~2 of \citet{kirkegaard2013incomplete}, who considers the special case in which the value distribution has a smooth, strictly positive density on a support starting at zero, so that the asymmetry effect vanishes at the first order. The local conclusion about weakening extends globally (Proposition~\ref{prop:global-weakening}): starting from symmetric bidders, any FOSD weakening of one bidder lowers expected revenue.

Section~\ref{sec:peace} applies these results to the literature on conflict preemption, in which an all-pay auction is preceded by a settlement stage \citep{zheng2019necessary,lu2026peacepayments}. A central step there is to identify the off-path posterior beliefs that make a deviating bidder most and least willing to reject a proposed settlement. According to the finding in that literature, the question boils down to which posterior belief minimizes the supremum $\bar{b}$ of the bidding support common to both bidders, and which posterior belief maximizes $\bar{b}$.  Theorem~\ref{theo:bid fosd} gives an immediate answer: among all posteriors supported on a fixed interval, $\bar b$ is smallest at the Dirac mass on the lower endpoint and largest at the Dirac mass on the upper endpoint, since these are the FOSD-extreme distributions. This recovers Lemma~2 of \citet{zheng2019necessary} and part of Lemma~10 of \citet{lu2026peacepayments}, both originally obtained through a more involved analysis of the bid distribution.

The paper is organized as follows. Section~\ref{model and equilibrium} derives the explicit equilibrium bid functions and their immediate corollaries. Section~\ref{sec:Bid distribution comparison} develops the FOSD comparative statics. Section~\ref{sec:revenue comparison} studies how expected revenue responds to changes in the value distributions. Section~\ref{sec:peace} presents the application to two-stage all-pay auctions. Section~\ref{conclusion} concludes. Technical details are in the Appendix.

\section{The Explicit Equilibrium}\label{model and equilibrium}
Consider a two-bidder all-pay auction for an indivisible prize. Bidders 1 and 2 have private valuations $v_1,v_2$, which are drawn from two independent distributions with c.d.f. $F_1,F_2$, respectively. Each $F_i$ is assumed to have bounded support $V_i\subseteq \mathbb{R}_+$, and to satisfy $F_i(0)=0$, i.e., neither bidder has a point mass on zero value.\footnote{As noted by \citet{olszewski2023equilibrium}, the presence of an atom at the value zero may lead to non-existence of equilibrium.}

After observing their realized values, each bidder simultaneously chooses a nonnegative bid. Both bidders pay their bids regardless of the winning status, and the bidder with the higher bid wins the prize. If the bids tie, we break the tie uniformly, i.e. each bidder wins with probability $\frac{1}{2}$. Thus, if bidders' values are $v_1,v_2$ and bids are $b_1,b_2$, bidder $i$'s ex post payoff is
\[
\begin{cases}
    v_i-b_i & b_i>b_j\\
    \frac{1}{2}v_i-b_i &b_i=b_j\\
    -b_i & b_i<b_j.
\end{cases}
\]

Denote a (mixed) strategy of bidder $i$ by $\Gamma_i:V_i\times\mathbb R\to[0,1]$, where $\Gamma_i(v_i,b)$ is the probability that bidder $i$ bids at most $b$ conditional on value $v_i$ (so $\Gamma_i(v_i,\cdot)$ is a c.d.f.\ for every $v_i\in V_i$). Denote by $\sigma_i(v_i)$ the probability measure induced by $\Gamma_i(v_i,\cdot)$. The strategy induces the unconditional bid distribution
\begin{equation}\label{eq:unconditional bid}
     G_i(b)=\int_{V_i}\Gamma_i(v_i,b)\,dF_i(v_i).
\end{equation}
For a profile $\Gamma=(\Gamma_1,\Gamma_2)$, let $BR_i(v_i)$ denote bidder $i$'s set of best responses at value $v_i$, given the opponent's bid distribution $G_{-i}$ induced by $\Gamma_{-i}$. The profile $\Gamma$ is a Bayesian Nash equilibrium (BNE) if, for $i=1,2$, $\sigma_i(v_i)$ assigns probability one to $BR_i(v_i)$ for $F_i$-almost every value $v_i$.

The following properties of no-gap, no-atom, common-bid-support, and monotonicity are standard in the all-pay auction literature. Proofs can be found in previous work \citep[see e.g.][]{amann1996asymmetric, siegel2014asymmetric, zheng2019necessary}.

\begin{lemma}[\citealt{amann1996asymmetric, siegel2014asymmetric, zheng2019necessary}]\label{lm:basic}
In any BNE of the two-bidder all-pay auction with independent values, the following necessary conditions hold:
    \begin{itemize}
\item[\textup{(i)}] \emph{No gap:}  if $G_i(a)=G_i(b)$ for $0<a<b$, then $G_i(a)=G_i(b)=1$.
\item[\textup{(ii)}] \emph{No atom:}  $G_i(b)$ is continuous at any $b>0$, and at most one bidder can have an atom at 0.
\item[\textup{(iii)}] \emph{Monotonicity:} For each bidder $i$, there exists an $F_i$-null set $N_i$ such that every type $v_i\notin N_i$ best responds, i.e.\ $\sigma_i(v_i)(BR_i(v_i))=1$, and the conditional bid supports $R_i(v_i):=\operatorname{supp}\sigma_i(v_i)$ are increasing in type: for all $v_i^1<v_i^2,v_i^1,v_i^2\notin N_i$, $\sup R_i(v_i^1)\le \inf R_i(v_i^2)$.
\item[\textup{(iv)}] \emph{Common bid support:} the unconditional bid distributions $G_1,G_2$ have a common support $[0,\bar{b}]$ for some $\bar{b}>0$.
    \end{itemize}
\end{lemma}

We use the quantile representation of values rather than the c.d.f.\ itself. For a distribution $F$, define its generalized inverse by\footnote{This quantile representation is similar in spirit to \citet{lu2017monotone}, who transform the signal into a uniform random variable on $[0,1]$, while the value becomes a non-trivial function of the transformed signal; in their notation, $\lambda_i(x,y)$ plays the role of the quantile function $Q_i(p)$ used here. The two frameworks are more general in different directions: they allow interdependent values and correlated signals, while this paper dispenses with the continuity of $Q_i$ ($\lambda_i$, in their notation) that their argument requires and delivers an explicit solution rather than an ODE characterization within the independent private value environment.}
\begin{equation}
    \label{inver-quantile}
    Q(p)\coloneq F^{-1}(p)\coloneq \inf\{v:F(v)\ge p\},
    \qquad p\in(0,1].
\end{equation}
The function $Q$ is called the quantile function in the rest of the paper. The value assigned to $Q(0)$ is immaterial, since the event  $\{p=0\}$ happens with probability zero; for definiteness, we set $Q(0)$ equal to the infimum of the support of the value distribution.

The original auction can be represented as the following \textit{quantile auction}. For each bidder $i$, let $p_i$ be independently drawn from the uniform distribution $U[0,1]$, and let bidder $i$'s value be $Q_i(p_i)$. By the definition of the generalized inverse, $Q_i(p_i)$ has distribution $F_i$. Hence the quantile auction is distributionally equivalent to the original auction. When $F_i$ has atoms, the interval of quantiles corresponding to a given atom provides a convenient way to represent the bidder's randomization conditional on that value.

Although a strategy in the quantile auction could still be mixed---bidder $i$ could randomize over bids after observing $p_i$, the following lemma shows that the construct of mixed strategy is not needed to capture the equilibria of the original game: every equilibrium of the original auction admits a representation in quantile space by a deterministic, monotone and continuous function.

\begin{figure}[h]
\centering
\begin{tikzpicture}[
    >=Stealth,
    every node/.style={font=\normalsize},
    arr/.style={->, thick, shorten >=2pt, shorten <=2pt},
    arrmix/.style={->, thick, dashed, shorten >=2pt, shorten <=2pt},
    lab/.style={font=\small, fill=white, inner sep=2pt}
]

\node (v) at (-4.5, 0) {value $v_i$};
\node (p) at (4.5, 0)  {quantile $p_i$};
\node (b) at (0, 3.8)  {bid $b$};

\draw[arrmix, transform canvas={yshift=6pt}] (v) --
    node[lab, above] {$F_i$, uniformize within atoms} (p);
\draw[arr, transform canvas={yshift=-6pt}] (p) --
    node[lab, below] {$v_i = Q_i(p_i)$} (v);

\draw[arrmix] (v) --
    node[lab, sloped, above] {mixed strategy $\sigma_i(v_i)$} (b);

\draw[arr] (p) --
    node[lab, sloped, above] {$b_i(p_i)$, deterministic} (b);
\end{tikzpicture}
\caption{Quantile representation of a strategy. Solid arrows are deterministic functions; dashed arrows involve randomization. Lemma~\ref{increasing bid function} asserts that the dashed strategy $v_i\mapsto\sigma_i(v_i)$ in the original auction admits an equivalent representation along the solid path $p_i\mapsto b_i(p_i)$, with the randomization absorbed into the $v_i\to p_i$ step.}
\label{fig:quantile-rep}
\end{figure}
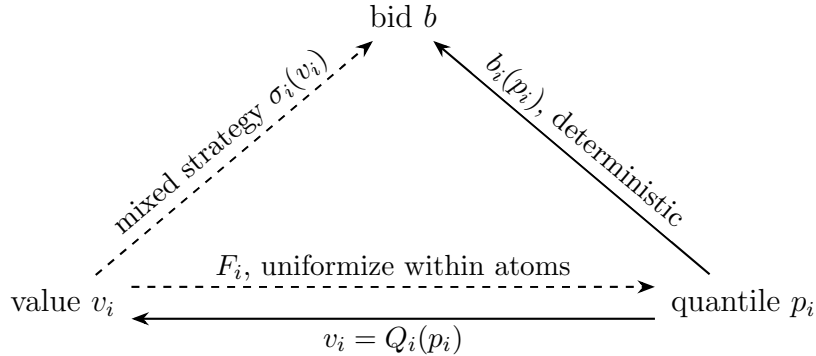

\begin{lemma}
    \label{increasing bid function}
Consider any Bayesian Nash equilibrium $(\sigma_1,\sigma_2)$ of the original auction (i.e., the auction in value space), possibly in mixed strategies. Let $G_i$ be the induced unconditional bid c.d.f.\ of bidder $i$, and define
    \begin{equation}\label{eq: lemma increasing bid}
        b_i(p)\coloneq G_i^{-1}(p)=\inf\{b\ge 0:G_i(b)\ge p\},\qquad p\in(0,1],
        \qquad b_i(0)\coloneq 0 .
    \end{equation}
Then $b_i:[0,1]\to[0,\bar b]$ is a deterministic, measurable function with the following two properties:
    \begin{enumerate}
\item[\textup{(i)}] \emph{Equivalence (a.e.):} for $p_i\sim U[0,1]$, the conditional distribution of $b_i(p_i)$ given $Q_i(p_i)=v_i$ coincides with $\sigma_i(v_i)$ for $F_i$-almost every value $v_i$. Equivalently, the joint distribution of $(Q_i(p_i),b_i(p_i))$ equals the equilibrium joint distribution of (value, bid) of bidder $i$. In particular, $(b_1,b_2)$ is a BNE of the quantile auction.
\item[\textup{(ii)}] \emph{Regularity:} each $b_i$ is continuous and weakly increasing on $[0,1]$, and strictly increasing on the set $\{p\in[0,1]:b_i(p)>0\}$. The only flat part of $b_i$ is a possible initial interval $[0,G_i(0)]$ (if $G_i(0)>0$) on which the bidder bids zero.
    \end{enumerate}
\end{lemma}

The proof is given in Appendix~\ref{Appendix: lemma increasing bid}. We briefly explain here how the lemma should be interpreted and used. The passage from value space to quantile space does nothing but rearrange the randomization that occurs at atoms (Figure~\ref{fig:quantile-rep}). When $F_i$ has an atom at $v_i$, the quantiles $p_i\in I_{v_i}=(F_i(v_i^-),F_i(v_i)]$ map to $v_i$.\footnote{Whether the left endpoint $p_i=F_i(v_i^-)$ satisfies $Q_i(p_i)=v_i$ depends on whether there is a gap to the left of $v_i$. However, a single point does not matter for the cumulative distribution, which is what we care about here.} The lemma shows that the bid $b_i(p_i)$, as $p_i$ ranges over $I_{v_i}$, reproduces the mixed bid of type $v_i$: $b_i(p_i)$ has distribution $\sigma_i(v_i)$ when $p_i$ is uniform on $I_{v_i}$. Because the value is constant on $I_{v_i}$, which quantile in $I_{v_i}$ draws which bid is payoff-irrelevant: the bids generated by $\sigma_i(v_i)$ may be assigned to the quantiles in $I_{v_i}$ in any order without changing the conditional bid distribution, hence without changing the equilibrium. Among these equivalent arrangements, defining $b_i=G_i^{-1}$ selects the increasing one.

Since $(b_1,b_2)$ reproduces the equilibrium joint distribution of value and bid, it induces the same unconditional bid distributions $(G_1,G_2)$ and hence leaves every best-response set unchanged. As it agrees with $(\Gamma_1,\Gamma_2)$ outside a null set, $(b_1,b_2)$ is a BNE of the quantile auction. Conversely, an equilibrium of the quantile auction can be directly mapped into an equilibrium of the original auction. So the two auctions are equivalent to analyze. In the rest of the paper, we analyze the quantile auction as it gives a cleaner result.

Define
\begin{equation}
\label{def S}
    S_i(p)\coloneq \int_p^1 \frac{dt}{Q_i(t)}\quad \forall p\in (0,1].
\end{equation}
Since $F_i(0)=0$, we have $Q_i(p)>0$ for all $p>0$. Hence $S_i$ is well-defined, finite, continuous, and strictly decreasing on $(0,1]$. It may diverge as $p\to0^+$, so we allow it to take values in the extended real line by defining
\[
S_i(0)\coloneq \lim_{p\to0^+}S_i(p)\in \mathbb{R}_+\cup \{+\infty\}.
\]
All comparisons involving $S_i(0)$ are understood in the extended real line. Since $S_i$ depends only on primitives and is defined symmetrically for the two bidders, we relabel bidders, if necessary, so that
\begin{equation}
\label{relabel}
    S_1(0)\geq S_2(0).
\end{equation}
The labels are fixed in this way throughout the rest of the paper. Notice that
\[
S_i(0) = \int_0^1 \frac{dt}{Q_i(t)}
       = \mathbb{E}\!\left[\frac{1}{v_i}\right]
\]
is the reciprocal of the harmonic mean of bidder $i$'s value distribution. Equivalently, the relabeling \eqref{relabel} sets bidder 2 to be the one with the (weakly) higher harmonic mean. Although the harmonic mean is a more natural primitive than its reciprocal, the analysis is cleaner in terms of $S_i$, which we therefore use throughout.\footnote{\label{footnote: marginal cost}In \citet{zheng2019necessary}, a bidder's type is interpreted as the reciprocal of the cost per unit of bid, so the ex post payoff of bidder $i$ is $\mathbf{1}\{i\text{ wins}\}-b_i/v_i$. This game is strategically equivalent to the all-pay auction considered here, since multiplying this payoff by the positive constant $v_i$ gives $v_i\mathbf{1}\{i\text{ wins}\}-b_i$. In \citeposs{zheng2019necessary} setup, $S_i(0)$ has a more direct economic interpretation: it is the expected cost per unit of bid for bidder $i$. Therefore, under \eqref{relabel}, bidder 2 is the bidder with the weakly lower expected cost ex ante.}

We are ready to present the main result of this paper, which provides an explicit formula for the equilibrium bid functions.
\begin{theo}\label{main theo}
In the two-bidder all-pay auction with bidders labeled according to \eqref{relabel}, there exists a BNE given by:
    \begin{equation}
    \label{bid function}
    b_1(p) =
    \begin{cases}
        0 & p\in[0,\alpha]\\
        \displaystyle\int_{\alpha}^p Q_2(\phi(s))\,ds & p\in (\alpha,1],
    \end{cases}
    \qquad
    b_2(p) = \int_0^p Q_1(\phi^{-1}(s))\,ds,
    \end{equation}
where
    \[
    \alpha =
    \begin{cases}
    0 & \text{if } S_1(0)=S_2(0) \\
    S_1^{-1}(S_2(0)) & \text{if } S_1(0)>S_2(0),
    \end{cases}
    \qquad
    \phi(p) =
    \begin{cases}
    0 & p \in [0, \alpha] \\
    S_2^{-1}(S_1(p)) & p \in (\alpha, 1],
    \end{cases}
    \]
and $\phi^{-1}$ denotes the inverse of $\phi$ restricted to $[\alpha, 1]$, on which $\phi$ is strictly increasing. Moreover, the equilibrium is unique in the sense that the induced joint distribution of $(v_i,b_i)$ is unique for each bidder in all possible equilibria.\footnote{Equilibrium existence in distributional strategies in the all-pay auction is already established by \citet{olszewski2023equilibrium} in a more general setup; the value of Theorem~\ref{main theo} lies in the explicit construction and uniqueness.}
\end{theo}
\begin{rmk}
The uniqueness statement is made at the level of the induced joint distribution of $(v_i,b_i)$. The formula in Theorem~\ref{main theo} selects the monotone quantile representative of this distribution. If $F_i$ has an atom at $v$, then all quantiles in the interval $(F_i(v^-),F_i(v)]$ correspond to the same value $v$. Thus a measure-preserving rearrangement of $b_i$ within this interval leaves the conditional distribution of bids given $v$, and hence the induced joint distribution of $(v_i,b_i)$, unchanged. In addition, changes on Lebesgue-null sets of quantiles are irrelevant under the a.e. best-response definition of equilibrium. Equivalently, in the original value representation, equilibrium strategies are unique up to changes on $F_i$-null sets of values.
\end{rmk}

We first check that the formula is well defined. Since each $S_i$ is continuous and strictly decreasing on $(0,1]$, it has a well-defined inverse on its range. If $S_1(0)>S_2(0)$, then necessarily $S_2(0)<+\infty$. Since $S_1(1)=0<S_2(0)<S_1(0)$, there is a unique $\alpha\in(0,1)$ such that $S_1(\alpha)=S_2(0)$. If $S_1(0)=S_2(0)$, then $\alpha=0$.

The function $\phi$ will be interpreted in the proof as the quantile matching function: bidder 1 with quantile $p$ and bidder 2 with quantile $\phi(p)$ submit the same bid in equilibrium, that is,
\begin{equation}
    \label{eq:matching}
    b_1(p)=b_2(\phi(p)).
\end{equation}
For every $p\in(\alpha,1]$, $S_1(p)$ belongs to the range of $S_2$, so $\phi(p)=S_2^{-1}(S_1(p))$ is well defined. Moreover, as $p\to\alpha^+$, we have $S_1(p)\to S_1(\alpha^+)=S_2(0^+)$, and hence $\phi(p)\to 0^+$. Thus the extension $\phi(\alpha)=0$ is continuous. Since both $S_1$ and $S_2$ are strictly decreasing, $\phi$ is strictly increasing on $[\alpha,1]$, with $\phi(\alpha)=0$ and $\phi(1)=1$.

The detailed proof is given in Appendix~\ref{appendix:proof main theo}. The following derivation explains the proof idea. Suppose for a moment that the functions $Q_i,b_i,\phi$ are smooth. Given bidder 2's equilibrium strategy, bidder 1 with quantile $p$ can index each deviation by the bidder 1 quantile $\tilde p$ that would submit the same bid in equilibrium. Thus bidder 1's problem can be written as
\begin{equation*}
    \max_{\tilde p\in[0,1]}
    \bigl\{Q_1(p)\phi(\tilde p)-b_1(\tilde p)\bigr\}
    =
    \max_{\tilde p\in[0,1]}
    \bigl\{Q_1(p)\phi(\tilde p)-b_2(\phi(\tilde p))\bigr\},
\end{equation*}
where the equality uses the matching identity \eqref{eq:matching}. At an interior equilibrium quantile $p\in(\alpha,1)$, the equilibrium choice is $\tilde p=p$. The first-order condition for the second maximization problem therefore gives
\[
    Q_1(p)=b_2'(\phi(p)).
\]
Similarly, bidder 2's first-order condition gives
\[
    Q_2(\phi(p))=b_1'(p).
\]
Differentiating \eqref{eq:matching} yields
\[
    b_1'(p)=b_2'(\phi(p))\phi'(p).
\]
Combining these identities gives
\[
    \frac{\phi'(p)}{Q_2(\phi(p))}=\frac{1}{Q_1(p)}.
\]
Using the boundary condition $\phi(1)=1$, this differential equation implies
\begin{equation}\label{eq:quantile}
    \int_{\phi(p)}^1 \frac{1}{Q_2(t)}dt
    =
    \int_p^1\frac{1}{Q_1(t)}dt
    \quad \forall p \in (\alpha,1],
\end{equation}
i.e., $S_2(\phi(p))=S_1(p)$ for all $p\in(\alpha,1]$.\footnote{For smooth distributions, a value-space counterpart of \eqref{eq:quantile} appears as equation~(7) of \citet{kirkegaard2008comparative}.} Thus $\phi=S_2^{-1}\circ S_1$, and integrating the first-order conditions recovers the bid functions in \eqref{bid function}. For general primitives, the functions need not be smooth. The proof therefore establishes the required absolute continuity and convexity, and derives the first-order conditions in an a.e.\ sense.

Equation \eqref{eq:quantile} has a useful cost interpretation. Under the reciprocal-cost interpretation of types in \citet{zheng2019necessary} (see footnote~\ref{footnote: marginal cost}), $1/Q_i(t)$ is the marginal cost of a unit of bid for the type at quantile $t$. Hence $S_i(p)$ is the aggregate marginal cost in the upper tail of bidder $i$'s type distribution above quantile $p$: since the quantile index is uniform, it equals the tail probability $1-p$ times the average marginal cost among quantiles above $p$. Thus the equilibrium matching rule $\phi$ does not match raw percentiles across bidders. Instead, it matches bidder 1's quantile $p$ with bidder 2's quantile $\phi(p)$ so that the two remaining upper tails have the same aggregate marginal cost.

Because $Q_i(p)>0$ for every $p>0$, \eqref{bid function} implies that $b_1$ is constant on $[0,\alpha]$ and strictly increasing on $(\alpha,1]$, while $b_2$ is strictly increasing on $[0,1]$. This immediately yields the following characterization of when the equilibrium is pure in the original value representation.
\begin{coro}
In the original value representation, the equilibrium is in pure strategies\footnote{A pure strategy means that each value type chooses a single bid.} if and only if bidder 1 has no value atom whose quantile interval intersects $(\alpha,1]$, and bidder 2 has no value atoms. In particular, if both value distributions are continuous, then the equilibrium is in pure strategies.
\end{coro}

Theorem~\ref{main theo} enables straightforward and direct computation of equilibrium bid strategies. For example, distributions with both a density part and mass points are, to the best of my knowledge, uncommon in the all-pay auction literature. By contrast, with the quantile method and Theorem~\ref{main theo}, such a mixture poses no additional difficulty, either conceptually or operationally. We illustrate this with a simple example.
\begin{eg}\label{eg:equilibrium calculation}
Consider the following auction. Bidder 1 has value $v_1\sim U[0,1]$. Bidder 2's value distribution has a continuous component of total mass $1/2$, uniform on $[0,1]$, and an atom of mass $1/2$ at $v=1$.

The quantile functions are
\[
Q_1(p)=p,\qquad 
Q_2(p)=
\begin{cases}
2p&0\le p\le 1/2\\
1&1/2<p\le 1.
\end{cases}
\]
According to \eqref{def S},
\[
S_1(p)=-\ln p, \qquad
S_2(p)=
\begin{cases}
-\frac12\ln(2p)+\frac12&0<p\le 1/2\\[4pt]
1-p&1/2<p\le 1.
\end{cases}
\]
Both $S_1(0)$ and $S_2(0)$ diverge, so $\alpha=0$. Solving $S_2(\phi(p))=S_1(p)$ gives
\[
\phi(p)=
\begin{cases}
\frac e2p^2&0\le p\le e^{-1/2}\\[4pt]
1+\ln p&e^{-1/2}<p\le 1.
\end{cases}
\]
According to \eqref{bid function}, the equilibrium bid functions are:
\[
b_1(p)=
\begin{cases}
\frac e3p^3&0\le p\le e^{-1/2}\\[4pt]
p-\frac{2}{3\sqrt e}&e^{-1/2}<p\le 1,
\end{cases}
\]
and
\[
b_2(p)=
\begin{cases}
\frac23\sqrt{\frac2e}\,p^{3/2}&0\le p\le 1/2\\[4pt]
e^{p-1}-\frac{2}{3\sqrt e}&1/2<p\le 1.
\end{cases}
\]
Equivalently, in the value space, bidder 1 bids:
\[
\beta_1(v)=
\begin{cases}
\frac e3v^3&0\le v\le e^{-1/2}\\[4pt]
v-\frac{2}{3\sqrt e}&e^{-1/2}<v\le 1,
\end{cases}
\]
bidder 2's continuous types $v_2\in[0,1)$ bid
\[
\beta_2(v_2)=\frac{v_2^{3/2}}{3\sqrt e},
\]
while the atom type $v_2=1$ mixes over the bid interval
\[
\left[\frac{1}{3\sqrt e},\,1-\frac{2}{3\sqrt e}\right]
\]
with density $g(b)=\frac{2}{b+\frac{2}{3\sqrt{e}}}$. Figure~\ref{fig:example1-revised} illustrates the quantile matching function and the resulting equilibrium bid functions, including how bidder 2's atom is represented in quantile space.
\end{eg}

\begin{figure}[!htbp]
\begin{singlespace}
\centering
\def\athresh{0.6065306597}

\begin{tikzpicture}

\begin{axis}[
  name=exonea,
  at={(0,0)},
  anchor=south west,
  width=0.40\linewidth,
  height=0.40\linewidth,
  scale only axis,
  xmin=0,xmax=1,
  ymin=0,ymax=1,
  axis lines=left,
  axis line style={-{Latex[length=3mm]}},
  xlabel={$p_1$},
  ylabel={$p_2$},
  xtick={0,0.6065306597,1},
  xticklabels={$0$,$e^{-1/2}$,$1$},
  ytick={0,0.5,1},
  yticklabels={$0$,$1/2$,$1$},
  tick label style={font=\footnotesize},
  label style={font=\small},
  tick align=outside,
  title={\textbf{(a) Quantile matching and allocation}},
  title style={font=\normalsize,yshift=0.6em},
  clip=false
]
  \path[fill=cOrange,fill opacity=0.13] (axis cs:0,0) rectangle (axis cs:1,1);

  \addplot[draw=none,fill=cBlue,fill opacity=0.16,domain=0:0.6065306597,samples=220,forget plot]
    {exp(1)/2*x^2} \closedcycle;
  \addplot[draw=none,fill=cBlue,fill opacity=0.16,domain=0.6065306597:1,samples=220,forget plot]
    {1+ln(x)} \closedcycle;

  \addplot[very thick,cGray!85!teal,domain=0:0.6065306597,samples=220,forget plot]
    {exp(1)/2*x^2};
  \addplot[very thick,cGray!85!teal,domain=0.6065306597:1,samples=220,forget plot]
    {1+ln(x)};

  \addplot[densely dashed,cGray!85,forget plot] coordinates {(0,0.5) (1,0.5)};
  \addplot[densely dashed,cGray!85,forget plot] coordinates {(0.6065306597,0) (0.6065306597,1)};

  \node[font=\footnotesize,cOrange!90!black] at (axis cs:0.70,0.90)
    {bidder 2 wins};
  \node[font=\footnotesize,cBlue!80!black] at (axis cs:0.76,0.22)
    {bidder 1 wins};

  \node[font=\footnotesize,align=center,anchor=center] (matchlabel)
    at (axis cs:0.40,0.72) {quantile matching function\\$p_2=\phi(p_1)$};
  \draw[-{Latex[length=2.0mm]},cGray!90]
    (axis cs:0.40,0.62) -- (axis cs:0.40,{exp(1)/2*0.40^2});
\end{axis}

\begin{axis}[
  name=exoneb,
  at={(0.54\linewidth,0)},
  anchor=south west,
  width=0.43\linewidth,
  height=0.40\linewidth,
  scale only axis,
  xmin=0,xmax=1,
  ymin=0,ymax=0.65,
  axis lines=left,
  axis line style={-{Latex[length=3mm]}},
  xlabel={quantile $p$},
  ylabel={bid},
  xtick={0,0.5,1},
  xticklabels={$0$,$1/2$,$1$},
  extra x ticks={0.6065306597},
  extra x tick labels={$e^{-1/2}$},
  extra x tick style={xticklabel style={xshift=11pt}},
  ytick={0,0.2,0.4,0.6},
  yticklabels={$0$,$0.2$,$0.4$,$0.6$},
  tick label style={font=\footnotesize},
  label style={font=\small},
  legend style={draw=none,fill=none,font=\footnotesize,at={(0.03,0.97)},anchor=north west},
  legend cell align=left,
  tick align=outside,
  title={\textbf{(b) Equilibrium bid functions}},
  title style={font=\normalsize,yshift=0.6em},
  clip=false
]
  \addplot[very thick,cBlue,domain=0:0.6065306597,samples=180]
    {exp(1)/3*x^3};
  \addlegendentry{$b_1(p)$}
  \addplot[very thick,cBlue,domain=0.6065306597:1,samples=180,forget plot]
    {x-2/(3*sqrt(exp(1)))};

  \addplot[very thick,cOrange,domain=0:0.5,samples=180]
    {2/3*sqrt(2/exp(1))*x^(3/2)};
  \addlegendentry{$b_2(p)$: continuous types}
  \addplot[very thick,cOrange,densely dashed,domain=0.5:1,samples=180]
    {exp(x-1)-2/(3*sqrt(exp(1)))};
  \addlegendentry{$b_2(p)$: atom $v_2=1$}

  \addplot[densely dotted,cGray!85,forget plot] coordinates {(0.5,0) (0.5,0.65)};
  \addplot[densely dotted,cGray!85,forget plot] coordinates {(0.6065306597,0) (0.6065306597,0.65)};
\end{axis}

\end{tikzpicture}

\caption{Example~\ref{eg:equilibrium calculation}. Panel (a) shows the quantile matching function $\phi$: bidder 1 wins below the matching function, and bidder 2 wins above it. Panel (b) shows the equilibrium bid functions. The dashed segment of $b_2$ on $p\in [1/2,1]$ represents the randomization of the atom at $v_2=1$ over the bid interval $[1/(3\sqrt{e}),1-2/(3\sqrt{e})]$.}
\label{fig:example1-revised}
\end{singlespace}
\end{figure}
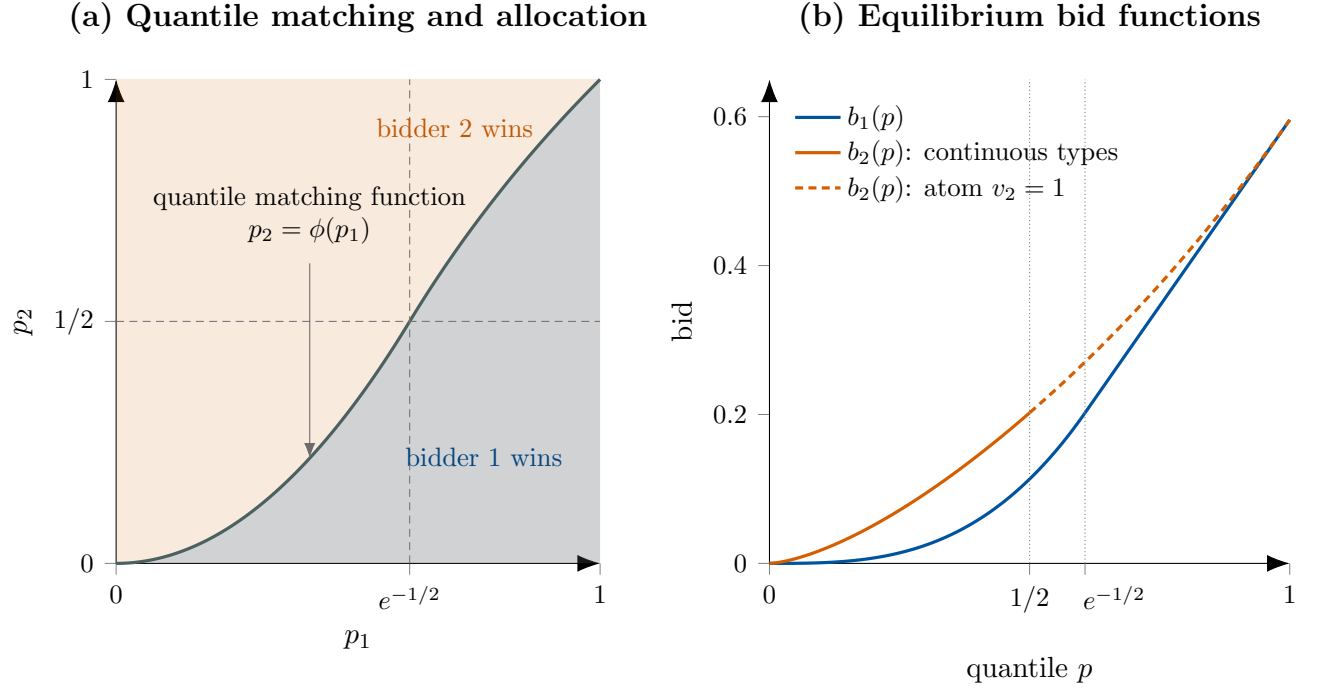

A corollary of Theorem~\ref{main theo} implies that in the two-bidder all-pay auction, the equilibrium allocation is ex post efficient only in the symmetric case. This contrasts with second-price auctions, where the equilibrium allocation is always efficient.

\begin{coro}\label{coro:efficient}
The equilibrium allocation is ex post efficient with probability one if and only if $F_1=F_2$.
\end{coro}
The ``if'' direction is standard. The contribution is the converse: although the literature contains many examples showing that asymmetric all-pay auctions can generate inefficient allocations, I am not aware of a general simple proof showing that symmetry is necessary for ex post efficiency under arbitrary value distributions. The proof uses the fact that efficiency requires $Q_2(\phi(p))=Q_1(p)$ for a.e.\ $p$. Combining this condition with the equilibrium ODE yields $\phi'(p)=1$ a.e.; the boundary condition then implies $\phi(p)=p$. The detailed proof is given in Appendix~\ref{Appendix:coro efficient}.

\section{Comparative Statics}
The explicit formula does more than enable numerical computation. Because the equilibrium bid functions are written directly in terms of the primitive quantile functions, comparative statics questions that would otherwise be hard to analyze become transparent. This section develops three such results. The first shows how a first-order stochastic dominance shift in a bidder's values moves the equilibrium bid distribution of that bidder (Section~\ref{sec:Bid distribution comparison}); the second derives expected revenue in closed form and studies local and global revenue comparisons around the symmetric benchmark (Section~\ref{sec:revenue comparison}); the third applies the first to the conflict preemption literature (Section~\ref{sec:peace}). The first generalizes comparative statics that \citet{kirkegaard2008comparative} established for smooth value distributions. The second is new, though the question it addresses goes back to \citet{kirkegaard2013incomplete}. The third provides simpler proofs of known results that originally required a more involved analysis.

\subsection{Intensity of Bidding}\label{sec:Bid distribution comparison}
We first ask how strengthening a bidder's value distribution affects that bidder's equilibrium bidding. For smooth value distributions, the answer is known from \citet{kirkegaard2008comparative}. The explicit formula extends it to arbitrary distributions and shortens the argument.
\begin{theo}
\label{theo:bid fosd}
Fix bidder $j$'s value distribution. If bidder $i$'s value distribution $\hat F_i$ FOSD (first-order stochastically dominates) $F_i$, then bidder $i$'s equilibrium bid distribution under $\hat F_i$ FOSD that under $F_i$.
\end{theo}
The formula in Theorem~\ref{main theo} makes the proof straightforward. A FOSD improvement of $F_i$ is equivalent to a pointwise increase in the quantile function $Q_i$. It can be shown that it shifts the matching function $\phi$ in a fixed direction, and the integral formula for $b_1$ then shifts bidder 1's bid function upward. The detailed proof is given in Appendix~\ref{Appendix: bid fosd}.

This result has a clean implication for the payoff of bidder $j$, whose value distribution is held fixed. Let $G_i$ and $\hat G_i$ denote bidder $i$'s equilibrium bid c.d.f.s under $F_i$ and $\hat F_i$, respectively, where $\hat F_i$ FOSD $F_i$. By Theorem~\ref{theo:bid fosd}, $\hat G_i\le G_i$ pointwise. Given bidder $i$'s bid distribution $G_i$, bidder $j$'s interim payoff at value $v$ is\footnote{In equilibrium, at most one bidder places an atom on zero bid, so ties occur with probability zero.}
\[
    \max_{b\ge0}\{vG_i(b)-b\}.
\]
Therefore,
\[
    v\hat G_i(b)-b\le vG_i(b)-b
    \qquad\text{for every }b,
\]
and taking maxima over $b$ shows that every type of bidder $j$ is weakly worse off facing the stronger opponent.

\begin{coro}\label{coro:opponent payoff}
Fix bidder $j$'s value distribution. If bidder $i$'s value distribution $\hat F_i$ first-order stochastically dominates $F_i$, then bidder $j$'s equilibrium interim payoff is weakly lower under $\hat F_i$ than under $F_i$ for every value $v$.
\end{coro}

The corollary is notable because strengthening $F_i$ causes both bidders' strategies to re-optimize. Nevertheless, Theorem~\ref{theo:bid fosd} orders bidder $i$'s induced bid distributions directly, so bidder $j$'s optimal value is ordered pointwise.

The closest results to Theorem~\ref{theo:bid fosd} and Corollary~\ref{coro:opponent payoff} in the literature are Proposition~4 and Corollary~2 of \citet{kirkegaard2008comparative}, which establish the same economic insights. \citet{kirkegaard2008comparative} studies value distributions that share a common lower support endpoint and are continuously differentiable with strictly positive densities. For a shift that strictly strengthens one bidder's distribution on the interior of its support (together with a boundary condition on the densities when the upper endpoint is unchanged), his Proposition~4 shows that the strengthened bidder's equilibrium bid distribution increases in the FOSD order, as in Theorem~\ref{theo:bid fosd}. His Corollary~2 shows that every type of the opponent that participates is strictly worse off, the counterpart of Corollary~\ref{coro:opponent payoff}. Relative to these results, the contribution here is generality. No regularity is imposed on the value distributions, and the shift is allowed to be weak. The proof of the theorem is also simpler, as the FOSD ranking of the bid distributions follows directly from the formula in Theorem~\ref{main theo}.\footnote{Related comparative statics for all-pay auctions appear in \citet{hopkins2007cross}, who study likelihood-ratio shifts in the common type distribution of symmetric bidders.}

A companion comparison holds across bidders within a single auction: if one bidder's values dominate the other's, the stronger bidder also bids higher for every quantile $p$.
\begin{prop}\label{prop:FOSD}
If bidder 2's value distribution FOSD bidder 1's, then bidder 2's bid distribution FOSD bidder 1's. Equivalently,
    \[
    Q_2\geq Q_1 \implies b_2(p)\geq b_1(p),\quad \forall p\in [0,1].
    \]
\end{prop}
\begin{proof}
When $Q_2\geq Q_1$, we have $S_1(p)\geq S_2(p)$. In particular, $S_1(0)\geq S_2(0)$, so the labeling is already consistent with \eqref{relabel} and Theorem~\ref{main theo}. Since $b_1(p)=b_2(\phi(p)),p\in [\alpha,1]$, it is equivalent to prove $\phi(p)\leq p$ for $p\in [\alpha,1]$.\footnote{When $p<\alpha$, $b_1(p)=0$, the conclusion is trivially true.} Suppose $\phi(p) > p$ for some $p$. Then
    \[
    S_1(p)=S_2(\phi(p))<S_2(p),
    \]
a contradiction.
\end{proof}

Theorem~\ref{theo:bid fosd} concerns the distribution of bids, not the bidder's payoff. Although Corollary~\ref{coro:opponent payoff} shows that the bidder whose distribution is held fixed is weakly worse off when facing a stronger opponent, the strengthened bidder need not benefit from the change. The natural conjecture that a bidder always benefits from becoming stronger is false: strengthening one bidder may intensify competition, as in the following example.
\begin{eg}\label{eg:bad strengthening}
Bidder 2 has value 2 for sure. For bidder 1, $Q_1(p)=p$ (uniform $[0,1]$) and $\hat Q_1(p)=1$ (value 1 for sure).\footnote{The degeneracy of $\hat Q_1$ is not essential. If the degenerate distribution at $1$ is replaced by a uniform distribution on $[1-\epsilon,1]$ with small $\epsilon>0$, bidder 1's interim payoff converges pointwise to the payoff under the degenerate distribution as $\epsilon\to 0$. For small $\epsilon$, some lower quantiles receive a small payoff gain of order $O(\epsilon)$ relative to the original $U[0,1]$ case, but sufficiently high quantiles are strictly worse off. Thus the example is robust rather than a knife-edge consequence of degeneracy.} Clearly $\hat Q_1\geq Q_1,\forall p$. However, under $\hat Q_1$, the expected payoff of bidder 1 is 0, since bidder 1 mixes between 0 and 1 with an atom at 0 in equilibrium.

Under $Q_1$, the equilibrium is given by Theorem~\ref{main theo}. Here $\alpha = e^{-1/2}$, and the equilibrium bid functions are:
\[
b_1(p)
=
\begin{cases}
0 & 0\le p\le \alpha\\[4pt]
2(p-\alpha) & \alpha<p\le 1,
\end{cases}
\]
\[
b_2(p)
=
\frac{2}{\sqrt e}\left(e^{p/2}-1\right)
\qquad 0\le p\le 1.
\]
The interim expected payoff of bidder 1 with quantile $p$ is given by
\[
U_1(p)
=
\begin{cases}
0 & 0\le p\le \alpha\\[6pt]
\dfrac{2}{\sqrt e}-p+2p\ln p & \alpha<p\le 1.
\end{cases}
\]
Thus bidder 1's equilibrium interim payoff is strictly positive whenever $v=p>e^{-1/2}$. Figure~\ref{fig:example2-payoff-bidder2-revised} contrasts the baseline payoff with the strengthened case and shows the corresponding increase in bidder 2's equilibrium bids.
\end{eg}

\begin{figure}[!htbp]
\begin{singlespace}
\centering
\def\astar{0.6065306597}

\begin{minipage}[t]{0.48\textwidth}
\centering
\parbox[t][2.4\baselineskip][t]{\linewidth}{\centering\textbf{(a) Bidder 1's interim payoffs}}\par\smallskip
\begin{tikzpicture}
\begin{axis}[
  width=0.96\linewidth,
  height=0.72\linewidth,
  xmin=0,xmax=1,
  ymin=0,ymax=1.05,
  axis lines=left,
  axis line style={-{Latex[length=3mm]}},
  xlabel={bidder 1 quantile $p$},
  ylabel={interim payoff},
  xtick={0,0.6065306597,1},
  xticklabels={$0$,$e^{-1/2}$,$1$},
  ytick={0,0.5,1},
  yticklabels={$0$,$0.5$,$1$},
  tick label style={font=\footnotesize},
  label style={font=\small},
  legend style={draw=none,fill=none,font=\scriptsize,at={(0.03,0.97)},anchor=north west,row sep=1pt},
  legend cell align=left,
  tick align=outside,
  clip=false
]
  \addplot[name path=oldpay,very thick,cBlue,domain=0.6065306597:1,samples=220,forget plot]
    {2/sqrt(exp(1))-x+2*x*ln(x)};
  \addplot[name path=base,draw=none,domain=0.6065306597:1,samples=2,forget plot] {0};
  \addplot[draw=none,fill=cBlue,fill opacity=0.15,forget plot] fill between[of=oldpay and base];

  \addplot[very thick,cBlue,domain=0:0.6065306597,samples=2,forget plot] {0};
  \addplot[very thick,cBlue,domain=0.6065306597:1,samples=220,forget plot]
    {2/sqrt(exp(1))-x+2*x*ln(x)};
  \addlegendimage{very thick,cBlue}
  \addlegendentry{baseline: $U_1(p)>0$ for $p>e^{-1/2}$}

  \addplot[very thick,cOrange,densely dashed,domain=0:1,samples=2,forget plot] {0};
  \addlegendimage{very thick,cOrange,densely dashed}
  \addlegendentry{strengthened: $\hat U_1(p)\equiv0$}

  \addplot[densely dotted,cGray!85,forget plot] coordinates {(0.6065306597,0) (0.6065306597,1.05)};
\end{axis}
\end{tikzpicture}
\end{minipage}\hfill%
\begin{minipage}[t]{0.48\textwidth}
\centering
\parbox[t][2.4\baselineskip][t]{\linewidth}{\centering\textbf{(b) Bidder 2's equilibrium bid functions}}\par\smallskip
\begin{tikzpicture}
\begin{axis}[
  width=0.96\linewidth,
  height=0.72\linewidth,
  xmin=0,xmax=1,
  ymin=0,ymax=1.05,
  axis lines=left,
  axis line style={-{Latex[length=3mm]}},
  xlabel={bidder 2 quantile $p$},
  ylabel={bid},
  xtick={0,0.5,1},
  xticklabels={$0$,$0.5$,$1$},
  ytick={0,0.5,1},
  yticklabels={$0$,$0.5$,$1$},
  tick label style={font=\footnotesize},
  label style={font=\small},
  legend style={draw=none,fill=none,font=\footnotesize,at={(0.03,0.97)},anchor=north west},
  legend cell align=left,
  tick align=outside,
  clip=false
]
  \addplot[very thick,cOrange,densely dashed,domain=0:1,samples=2]
    {x};
  \addlegendentry{strengthened: $\hat b_2(p)=p$}
  \addplot[very thick,cBlue,domain=0:1,samples=220]
    {2/sqrt(exp(1))*(exp(x/2)-1)};
  \addlegendentry{baseline: $b_2(p)$}
\end{axis}
\end{tikzpicture}
\end{minipage}

\caption{Example~\ref{eg:bad strengthening}. Strengthening bidder 1 from $Q_1(p)=p$ to $\hat Q_1(p)=1$ eliminates bidder 1's positive rents at the top quantiles because the strengthening makes bidder 2 bid more aggressively: $b_2(p)=\frac{2}{\sqrt e}(e^{p/2}-1)$ increases to $\hat b_2(p)=p$.}
\label{fig:example2-payoff-bidder2-revised}
\end{singlespace}
\end{figure}
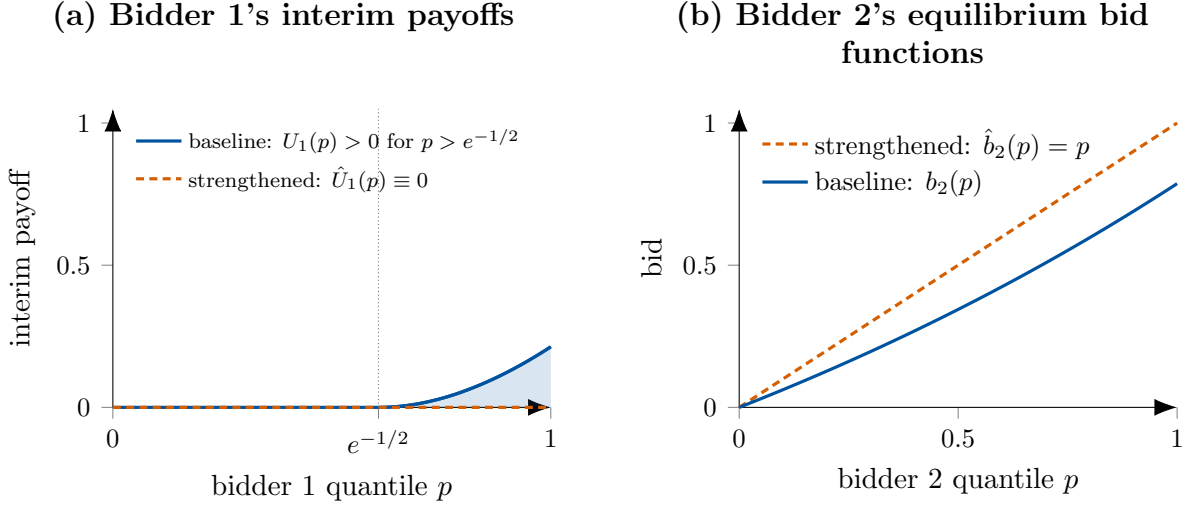

\subsection{Revenue Comparison}\label{sec:revenue comparison}
Next, we turn from the distribution of bids to expected revenue. The central question of this subsection is how expected revenue responds when one of two symmetric bidders becomes slightly stronger or weaker. A change in one bidder's value distribution has two effects. First, it changes the overall level of values in the auction, and higher values support higher bids. We call this the level effect. Second, the change in either direction makes the auction asymmetric. We call this the asymmetry effect. The level effect follows the direction of the change, while the asymmetry effect lowers expected revenue in both directions. The two effects therefore reinforce each other when a bidder is weakened and conflict when that bidder is strengthened. The analysis below identifies when the asymmetry effect matters at the first order and through which channel.

With the help of Theorem~\ref{main theo}, expected revenue can also be written in closed form. By Fubini's theorem, integrating \eqref{bid function} gives
\[
    \mathbb{E}[b_2]
    =
    \int_0^1 b_2(p)\,dp
    =
    \int_0^1 (1-s)Q_1(\phi^{-1}(s))\,ds.
\]
The change of variables $s=\phi(t)$ further simplifies this to
\[
    \int_\alpha^1 (1-\phi(t))Q_1(t)\phi'(t)\,dt
    =
    \int_\alpha^1 (1-\phi(t))Q_2(\phi(t))\,dt,
\]
where the last equality follows from \eqref{ode}. Similarly,
\[
    \mathbb{E}[b_1]
    =
    \int_\alpha^1 (1-s)Q_2(\phi(s))\,ds.
\]
The expected revenue can be expressed as
\begin{equation}
\label{eq:revenue1}
    \mathrm{Rev}(Q_1,Q_2)
    =
    \int_\alpha^1 (2-s-\phi(s))Q_2(\phi(s))\,ds.
\end{equation}
Equivalently, after the change of variables $t=\phi(s)$,
\begin{equation}
\label{eq:revenue}
    \mathrm{Rev}(Q_1,Q_2)
    =
    \int_0^1 (2-t-\phi^{-1}(t))Q_1(\phi^{-1}(t))\,dt.
\end{equation}
Since a change in the value distributions changes $Q_i$, $\phi$, and $\alpha$ simultaneously, revenue comparisons are generally ambiguous. We therefore study local FOSD perturbations around the symmetric benchmark. Call a function $H:[0,1]\to\mathbb R_+$ an admissible perturbation direction at $Q$ if there exists $\bar\epsilon>0$ such that $Q+\epsilon H$ is a quantile function for every $\epsilon\in(-\bar\epsilon,\bar\epsilon)$.\footnote{The condition is two-sided: for all sufficiently small $\epsilon>0$, both $Q+\epsilon H$ and $Q-\epsilon H$ must be nonnegative, nondecreasing, and left-continuous.}

\begin{lemma}[Directional revenue derivatives]\label{lem:revenue derivatives}
Let $Q$ be a fixed quantile function, and let $H$ be an admissible perturbation direction at $Q$. Define
    \[
        R_H(\epsilon):=\mathrm{Rev}(Q,Q+\epsilon H).
    \]
Then the one-sided derivatives of $R_H$ at $\epsilon=0$ are
    \[
        (R_H)'_+(0)=I_H-B_H-J_H,
        \qquad
        (R_H)'_-(0)=I_H+B_H+J_H,
    \]
where
    \[
        I_H=\int_0^1(1-p)H(p)\,dp,\quad
        B_H=\lim_{p\to0^+}T_H(p)Q(p)^2,\quad
        J_H=\sum_{p\in(0,1)}(1-p)T_H(p)(\Delta Q(p))^2,
    \]
with
    \[
        T_H(p)=\int_p^1\frac{H(t)}{Q(t)^2}\,dt.
    \]
Here $\Delta Q(p)=Q(p^+)-Q(p)$ denotes the jump of $Q$ at $p$. Moreover, $B_H=0$ when $Q(0)=0$, and $B_H=T_H(0)Q(0)^2$ when $Q(0)>0$.
\end{lemma}

The proof consists of direct computation with the equilibrium formula in Theorem~\ref{main theo} and the revenue expression \eqref{eq:revenue}. Technical details are given in Appendix~\ref{proof: marginal revenue}.

The admissibility requirement is not automatic. It means that $H$ is a two-sided perturbation direction within the cone of quantile functions.\footnote{In particular, two-sided admissibility implies that $H$ is left-continuous and, for some $C<\infty$,
\[
    0\le H(p)\le C Q(p)
    \qquad\text{for all }p\in[0,1],
\]
and
\[
    |H(q)-H(p)|
    \le
    C\bigl(Q(q)-Q(p)\bigr)
    \qquad\text{for all }p<q.
\]
Thus $H$ must be constant on every interval on which $Q$ is constant.} Several important perturbations satisfy the requirement. (i) If $H$ is constant, then the perturbation is an additive shift, and admissibility holds whenever $Q(0)>0$. (ii) If $H=Q$, then the perturbation is a scaling of the value distribution, and admissibility always holds. (iii) If $Q$ is bounded away from zero and satisfies a lower Lipschitz condition $Q(q)-Q(p)\ge m(q-p),\forall q>p$ for some $m>0$, while $H$ is Lipschitz, then the perturbation is admissible.

Lemma~\ref{lem:revenue derivatives} has direct implications for how expected revenue responds locally to perturbations around the symmetric benchmark.
\begin{prop}
\label{prop:local revenue effects}
Let $Q$ be a fixed quantile function, and let $H$ be an admissible perturbation direction at $Q$ that is non-trivial, i.e.\ $H>0$ on a set of positive Lebesgue measure. Then:
    \begin{enumerate}
\item[\textup{(i)}] Local weakening is harmful to expected revenue: $(R_H)'_-(0)>0$. Equivalently, for all small $\epsilon>0$, $R_H(-\epsilon)<R_H(0)$.
\item[\textup{(ii)}] The sign of the local strengthening effect is ambiguous in general, i.e.\ $(R_H)'_+(0)$ can be positive or negative.
\item[\textup{(iii)}] The marginal revenue loss from weakening is weakly larger than the marginal revenue effect of the corresponding strengthening:
        \[
\begin{aligned}
    &\lim_{\epsilon\to0^+}
    \frac{R_H(0)-R_H(-\epsilon)}{\epsilon}
    -
    \lim_{\epsilon\to0^+}
    \frac{R_H(\epsilon)-R_H(0)}{\epsilon}  \\
    &\qquad
    =
    (R_H)'_-(0)-(R_H)'_+(0)
    =
    2(B_H+J_H)\ge0.
\end{aligned}
\]
    \end{enumerate}
\end{prop}
The proof follows immediately from Lemma~\ref{lem:revenue derivatives}. For part \textup{(i)}, since $I_H>0$ and $B_H,J_H\ge 0$, $(R_H)'_-(0)=I_H+B_H+J_H>0$. For part \textup{(ii)}, if $Q$ is continuous and $Q(0)=0$, then $B_H=J_H=0$, so $(R_H)'_+(0)=I_H>0$, and strengthening raises expected revenue locally. On the other hand, consider the example of $Q=H\equiv1$. Then $I_H=1/2$, $B_H=1$, and $J_H=0$, so $(R_H)'_+(0)=-1/2<0$, and strengthening lowers expected revenue locally. Thus the symmetric benchmark has a one-sided revenue-favorability property: weakening is always locally harmful, whereas strengthening need not be beneficial.

Part \textup{(iii)} follows by subtracting the right derivative from the left derivative. Moreover, the inequality is strict when $B_H+J_H>0$. The boundary term $B_H$ is strictly positive whenever $Q(0)>0$, which corresponds to the lower endpoint of the value support being bounded away from zero. The jump term $J_H$ is strictly positive whenever $Q$ has a jump at some $p_0\in(0,1)$, which corresponds to a gap in the value distribution, and $H$ is positive on a set of positive measure in $[p_0,1]$. Thus the expected revenue function may have a concave kink at the symmetric benchmark.

The decomposition in Lemma~\ref{lem:revenue derivatives} formalizes the two effects described at the beginning of this subsection. The term $I_H$ is the level effect. It follows the direction of the change and enters the two one-sided derivatives with the same sign. The sum $B_H+J_H$ is the asymmetry effect. It weakly lowers expected revenue in both directions of the change and therefore enters the two derivatives with opposite signs. The boundary term $B_H$ is a participation margin at the bottom: the perturbation pushes the weaker bidder's lowest types into bidding zero. Since the initial slope of the bid function in \eqref{bid function} vanishes precisely when $Q(0)=0$, the displaced bids matter at the first order when $Q(0)>0$. The jump term $J_H$ shows that an interior gap in the value distribution creates an analogous margin, since the slope of the bid function changes discretely at the corresponding quantile.

Proposition~\ref{prop:local revenue effects} generalizes Proposition~2 of \citet{kirkegaard2013incomplete}. Kirkegaard studies the same local revenue question along the scaling direction, which corresponds to $H=Q$ in our notation. Under his regularity assumptions, the value distribution has a $\mathcal C^1$, strictly positive density on a support of the form $[0,\bar v]$. In our notation, this implies $Q(0)=0$ and rules out jumps of $Q$. Hence $B_Q=J_Q=0$, and the two one-sided derivatives coincide. Equivalently, expected revenue varies smoothly through the symmetric benchmark along the scaling path with a positive derivative $I_Q$. Lemma~\ref{lem:revenue derivatives} shows how this smoothness can fail outside that particular environment: the boundary term $B_H$ and the jump term $J_H$ create the derivative gap
\[
    (R_H)'_-(0)-(R_H)'_+(0)=2(B_H+J_H).
\]
In fact, the local conclusion that weakening one bidder from the symmetric benchmark reduces expected revenue has a stronger global counterpart: starting from symmetric bidders, any non-trivial weakening of one bidder strictly lowers expected revenue.
\begin{prop}\label{prop:global-weakening}
Let $Q$ and $W$ be bounded quantile functions with $W\leq Q$, satisfying the maintained assumption that $W(p),Q(p)>0$ whenever $p>0$. Then
    \[
        \mathrm{Rev}(Q,Q)\geq \mathrm{Rev}(W,Q).
    \]
The inequality is strict whenever $W<Q$ on a set of positive Lebesgue measure.
\end{prop}
The proof is given in Appendix~\ref{Appendix:global weaken}.

The global weakening result is special to the symmetric benchmark. A tempting extension would be to conjecture that, from any asymmetric benchmark, further weakening the already weaker bidder must lower expected revenue. This conjecture is false. Once bidders are already asymmetric, weakening the weaker bidder may change the quantile matching function $\phi$ in a way that can raise the stronger bidder's induced payments enough to dominate the direct loss from the weaker value distribution.

\begin{eg}\label{ex:weakening-weaker-raises-revenue}
Let bidder 2 have quantile function $Q_2(p)\equiv1$. For $x\in(1/2,1]$, let bidder 1 have quantile function
    \[
        Q_x(p)=
        \begin{cases}
            1/6, & 0\le p\le 1/2,\\
            x, & 1/2<p\le 1.
        \end{cases}
    \]
Then $Q_x\le Q_2$ for every $x\in(1/2,1]$, so bidder 1 is the weaker bidder. Let
    \[
        R(x):=\mathrm{Rev}(Q_x,Q_2).
    \]
Then $R(x)$ is strictly decreasing in $x$ on $(1/2,1]$. Hence lowering $x$, which weakens bidder 1, raises expected revenue. The detailed computation is given in Appendix~\ref{appendix:example}.
\end{eg}

This example clarifies the scope of the preceding results. Symmetry has a special revenue-favorability property: any non-trivial weakening from the symmetric benchmark strictly lowers expected revenue, locally and globally. However, expected revenue is not monotone in the weaker bidder's strength at arbitrary asymmetric benchmarks.

\subsection{The Most and Least Penalizing Off-Path Posteriors}\label{sec:peace}
Theorem~\ref{theo:bid fosd}, which says a stronger bidder bids higher in the FOSD order, can be applied to provide simpler proofs of comparative statics observations that are central to the literature on conflict preemption, where an all-pay auction is preceded by a settlement negotiation stage. We recast the key steps of \citet{zheng2019necessary} and \citet{lu2026peacepayments} in our notation.

In their models, either one bidder or a neutral mediator first proposes a peaceful settlement. Specifically, in \citet{zheng2019necessary}, a neutral mediator proposes a split of the prize; in \citet{lu2026peacepayments}, one bidder proposes a bribe to the other. If it fails to settle, the two bidders play the all-pay auction with updated posterior beliefs. \citet{zheng2019necessary} and \citet{lu2026peacepayments} study the peaceful equilibrium, in which the failure of settlement is an off-path event, hence the posterior beliefs in principle can be any distribution supported on the original support.

These off-path beliefs are unrestricted, and the literature distinguishes two notions of conflict preemption according to how such a belief is resolved. Peace is \textit{implementable} if some off-path belief sustains it as an equilibrium, reflecting a mediator who can coordinate the bidders on a favorable continuation play. Peace is \textit{securable} if acceptance remains optimal under every off-path belief. To study either notion, one compares each type's expected payoff from accepting the proposal against its equilibrium payoff from rejection. The type for which the incentive constraint binds is the highest type, since this type gets the highest payoff in the all-pay auction stage among all possible types. In the equilibrium of the all-pay auction, one of this type’s best responses is to bid the common supremum $\bar b$ of the bid support and win with probability one, yielding a rejection payoff of $\bar v_i-\bar b$.\footnote{\citet{zheng2019necessary} uses the strategically equivalent normalized payoff $\mathbbm{1}\{i\text{ wins}\} - b_i/v_i$, obtained from the standard all-pay payoff by dividing by $v_i$.} The implementability and security questions therefore reduce to how $\bar b$ varies with the posterior $F_i$. This becomes simple to analyze once we have Theorem~\ref{theo:bid fosd} in hand.

Let $\mathrm{supp}(F_i^0)$ denote the support of bidder $i$'s original value distribution, and let $\underline v_i$ and $\bar v_i$ be its smallest and largest points (equivalently, the endpoints of its convex hull). Write $\bar b(F_i)$ for the supremum of the equilibrium bid support when bidder $i$'s posterior distribution is $F_i$ and bidder $j$'s is held fixed.
\begin{coro}\label{coro:bbar bounds}
Fix bidder $j$'s value distribution $F_j$. Among all posteriors $F_i$ with support contained in $\mathrm{supp}(F_i^0)$, the supremum of the bid support $\bar b(F_i)$ is minimized at $\delta_{\underline v_i}$\footnote{If $\underline v_i=0$, the statement should be interpreted as a limiting result as the posterior degenerates to values approaching zero.} and maximized at $\delta_{\bar v_i}$, where $\delta_x$ is the Dirac measure at $x$.
\end{coro}
\begin{proof}
Every such $F_i$ satisfies $\delta_{\bar v_i}\succeq_{\mathrm{FOSD}}F_i \succeq_{\mathrm{FOSD}}\delta_{\underline v_i}$. By Theorem~\ref{theo:bid fosd}, bidder $i$'s equilibrium bid distribution inherits this FOSD ordering, hence so does its supremum $\bar b$.
\end{proof}

Corollary~\ref{coro:bbar bounds} delivers essentially the same result as Lemma~2 of \citet{zheng2019necessary} and the first inequality of Lemma~10 of \citet{lu2026peacepayments}.\footnote{The second inequality of Lemma~10 in \citet{lu2026peacepayments} states that bidder $j$'s set of rejection types is largest when $F_i=\delta_{\underline v_i}$. This also follows from Theorem~\ref{theo:bid fosd}: since $\delta_{\underline v_i}$ induces the weakest bid distribution for bidder $i$, every type of bidder $j$ obtains a higher expected auction payoff and is therefore more inclined to reject.} \citet{zheng2019necessary} proves the comparison by analyzing how the equilibrium bid-support supremum varies with the off-path posterior. His argument uses a distributional characterization of the all-pay equilibrium and compares the induced bid distributions through the linkage between marginal costs, marginal revenues, and the slopes of equilibrium bid distributions.\footnote{In the notation of \citet{zheng2019necessary}, $F_i^{-1}$ denotes the generalized inverse of the type distribution, so it corresponds to our quantile function $Q_i$. His bid-to-type map $\gamma_{i,\sigma}(b)=\widetilde F_i^{-1}(H_{i,\sigma}(b))$ is the bid-space counterpart: when $p=H_{i,\sigma}(b)$, one has $\gamma_{i,\sigma}(b)=Q_i(p)$ for the relevant type distribution. Thus the reciprocal marginal-cost term $1/\gamma_{i,\sigma}(b)$ in his bid-space analysis corresponds to $1/Q_i(p)$ in our quantile-space formulation.} \citet{lu2026peacepayments} rely on \citeposs{zheng2019necessary} result for this part. With Theorem~\ref{theo:bid fosd}, the comparison of suprema becomes an immediate implication of the FOSD ranking.

\section{Conclusion}\label{conclusion}

This paper provides an explicit solution to the two-bidder all-pay auction with asymmetric independent private values. The main result expresses the equilibrium bid functions directly in terms of the primitives, without imposing continuity or discreteness assumptions on the value distributions. The formula brings these cases under a common framework and makes the equilibrium directly applicable for further analysis.

The other results illustrate why this explicitness is useful. First, the formula yields first-order stochastic dominance comparative statics for equilibrium bid distributions: when one bidder's value distribution increases in the FOSD order, the corresponding equilibrium bid distribution also increases in the FOSD order. This comparison, previously known for smooth value distributions, holds here without any regularity assumption, and implies that the opponent is weakly worse off against a stronger bidder. Second, taking expectations of the bid functions gives a closed-form expression for expected revenue. This generalizes local revenue comparisons around the symmetric benchmark by allowing arbitrary admissible perturbation directions and by identifying how boundary and jump terms create a gap between the two one-sided marginal effects. It also yields a global comparison: starting from symmetry, weakening one bidder lowers expected revenue. Third, the formula simplifies central arguments in some multi-stage all-pay auction environments. In the application to conflict preemption, the FOSD comparative statics gives a direct way to identify the posterior beliefs that minimize or maximize the relevant continuation payoff, avoiding more involved derivations of the equilibrium bid distribution.

Methodologically, the paper also shows the usefulness of working in quantile space. The quantile representation absorbs the randomization generated by atoms in the value distribution and allows equilibrium bid functions to be represented as deterministic and monotone functions of quantiles. This representation makes it possible to recover the first-order condition behind the classical ODE approach almost everywhere and then integrate it directly. In this sense, the closed-form formula can be viewed as an explicit integration of the all-pay auction ODE in quantile space, but one that remains valid without the regularity assumptions required by the standard ODE approach.

The main limitation of the analysis is the independent private value assumption. This paper relies on both the independence and the private-value parts of this assumption, but in different ways. Independence is used in the proof because each type of a bidder faces the same opponent bid distribution. With correlated private values, the opponent's conditional bid distribution may vary with the bidder's own type, so the convexity argument behind Theorem~\ref{main theo} (Lemma~\ref{Lem:conv} in the appendix) need not apply directly. This difficulty concerns the proof method rather than necessarily the form of the result. A more substantial obstacle arises with interdependent values. In that case, a bidder's payoff from a bid depends not only on the probability of winning, but also on the value conditional on the opponent types associated with that bid. As a result, the first-order conditions no longer yield a separable ODE, and the integration argument that delivers the explicit formula breaks down.\footnote{\citet{lu2017monotone} show that signal correlation can be absorbed into a payoff-equivalent fictitious auction with independent uniform signals and suitably transformed valuation functions. The transformed valuation functions generally remain interdependent, and this is precisely what prevents the first-order conditions here from reducing to a separable ODE.} Extending the primitive-based approach beyond independent private values is therefore a natural question for future research.

More broadly, the formula developed here may be useful in applications where all-pay auctions appear as continuation games and where posterior beliefs need not be smooth or finite. Information disclosure, information design, and multi-stage conflict models often generate such posterior distributions endogenously. The results of this paper suggest that explicit primitive-based characterizations can make these environments more tractable, both by simplifying equilibrium computation and by making comparative statics transparent.

\bibliography{ref}

@article{siegel2014asymmetric,
  title={Asymmetric all-pay auctions with interdependent valuations},
  author={Siegel, Ron},
  journal={Journal of Economic Theory},
  volume={153},
  pages={684--702},
  year={2014},
  publisher={Elsevier}
}

@article{amann1996asymmetric,
  title={Asymmetric all-pay auctions with incomplete information: the two-player case},
  author={Amann, Erwin and Leininger, Wolfgang},
  journal={Games and economic behavior},
  volume={14},
  number={1},
  pages={1--18},
  year={1996},
  publisher={Elsevier}
}

@article{lu2017monotone,
  title={Monotone equilibrium of two-bidder all-pay auctions Redux},
  author={Lu, Jingfeng and Parreiras, S{\'e}rgio O},
  journal={Games and Economic Behavior},
  volume={104},
  pages={78--91},
  year={2017},
  publisher={Elsevier}
}

@article{olszewski2023equilibrium,
  title={Equilibrium existence in games with ties},
  author={Olszewski, Wojciech and Siegel, Ron},
  journal={Theoretical Economics},
  volume={18},
  number={2},
  pages={481--502},
  year={2023},
  publisher={Wiley Online Library}
}

@article{zheng2019necessary,
  title={Necessary and sufficient conditions for peace: Implementability versus security},
  author={Zheng, Charles Z},
  journal={Journal of Economic Theory},
  volume={180},
  pages={135--166},
  year={2019},
  publisher={Elsevier}
}

@book{leoni2017first,
  title={A first course in Sobolev spaces},
  author={Leoni, Giovanni},
  year={2017},
  publisher={American Mathematical Soc.}
}

@article{kirkegaard2013incomplete,
  title={Incomplete information and rent dissipation in deterministic contests},
  author={Kirkegaard, Ren{\'e}},
  journal={International Journal of Industrial Organization},
  volume={31},
  number={3},
  pages={261--266},
  year={2013},
  publisher={Elsevier}
}

@article{hopkins2007cross,
  title={Cross and double cross: Comparative statics in first price and all pay auctions},
  author={Hopkins, Ed and Kornienko, Tatiana},
  journal={The B.E. Journal of Theoretical Economics},
  volume={7},
  number={1},
  pages={19},
  year={2007},
  publisher={De Gruyter}
}

@article{kirkegaard2008comparative,
  author  = {Kirkegaard, Ren{\'e}},
  title   = {Comparative Statics and Welfare in Heterogeneous All-Pay Auctions:
             Bribes, Caps, and Performance Thresholds},
  journal = {The B.E. Journal of Theoretical Economics},
  volume  = {8},
  number  = {1},
  pages   = {1--32},
  year    = {2008}
}

@article{chen2025regularization,
  title={The Regularization and Optimal Design of All-Pay Auctions},
  author={Chen, Bo and Lu, Jingfeng and Bai, Pengcheng},
  journal={Available at SSRN 5314593},
  year={2025}
}

@misc{lu2026peacepayments,
      title={Peace Through Side Payments}, 
      author={Jingfeng Lu and Zongwei Lu and Christian Riis},
      year={2026},
      eprint={2107.11575},
      archivePrefix={arXiv},
      primaryClass={econ.TH},
      url={https://arxiv.org/abs/2107.11575}, 
}

@article{morath2008private,
  title={Private versus complete information in auctions},
  author={Morath, Florian and M{\"u}nster, Johannes},
  journal={Economics Letters},
  volume={101},
  number={3},
  pages={214--216},
  year={2008},
  publisher={Elsevier}
}

@article{lu2018ranking,
  title={Ranking disclosure policies in all-pay auctions},
  author={Lu, Jingfeng and Ma, Hongkun and Wang, Zhe},
  journal={Economic Inquiry},
  volume={56},
  number={3},
  pages={1464--1485},
  year={2018},
  publisher={Wiley Online Library}
}

@article{chen2017persuasion,
  title={Persuasion and timing in asymmetric-information all-pay auction contests},
  author={Chen, Jian and Kuang, Zhonghong and Zheng, Jie},
  journal={Available at SSRN 4099339},
  year={2017}
}

@article{bulow1989simple,
  title={The simple economics of optimal auctions},
  author={Bulow, Jeremy and Roberts, John},
  journal={Journal of political economy},
  volume={97},
  number={5},
  pages={1060--1090},
  year={1989},
  publisher={The University of Chicago Press}
}

@inproceedings{alaei2013simple,
  title={The simple economics of approximately optimal auctions},
  author={Alaei, Saeed and Fu, Hu and Haghpanah, Nima and Hartline, Jason},
  booktitle={2013 IEEE 54th Annual Symposium on Foundations of Computer Science},
  pages={628--637},
  year={2013},
  organization={IEEE}
}

@article{maskin2000asymmetric,
  title={Asymmetric auctions},
  author={Maskin, Eric and Riley, John},
  journal={The review of economic studies},
  volume={67},
  number={3},
  pages={413--438},
  year={2000},
  publisher={Wiley-Blackwell}
}

@article{lebrun1999first,
  title={First price auctions in the asymmetric N bidder case},
  author={Lebrun, Bernard},
  journal={International Economic Review},
  volume={40},
  number={1},
  pages={125--142},
  year={1999},
  publisher={Wiley Online Library}
}

@article{milgrom1985distributional,
  title={Distributional strategies for games with incomplete information},
  author={Milgrom, Paul R and Weber, Robert J},
  journal={Mathematics of operations research},
  volume={10},
  number={4},
  pages={619--632},
  year={1985},
  publisher={INFORMS}
}

@article{fu2014disclosure,
  title={Disclosure policy in a multi-prize all-pay auction with stochastic abilities},
  author={Fu, Qiang and Jiao, Qian and Lu, Jingfeng},
  journal={Economics Letters},
  volume={125},
  number={3},
  pages={376--380},
  year={2014},
  publisher={Elsevier}
}

@article{prokopovych2023monotone,
  title={On monotone pure-strategy Bayesian-Nash equilibria of a generalized contest},
  author={Prokopovych, Pavlo and Yannelis, Nicholas C},
  journal={Games and Economic Behavior},
  volume={140},
  pages={348--362},
  year={2023},
  publisher={Elsevier}
}

@book{nelsen2006introduction,
  title={An introduction to copulas},
  author={Nelsen, Roger B},
  year={2006},
  publisher={Springer}
}

@article{rentschler2016two,
  title={Two-bidder all-pay auctions with interdependent valuations, including the highly competitive case},
  author={Rentschler, Lucas and Turocy, Theodore L},
  journal={Journal of Economic Theory},
  volume={163},
  pages={435--466},
  year={2016},
  publisher={Elsevier}
}

\appendix
\numberwithin{equation}{section}

\section{Omitted proofs}

\subsection{Proof of Lemma~\ref{increasing bid function}}\label{Appendix: lemma increasing bid}

Throughout we fix a BNE and, since the two bidders are analyzed separately, suppress the index $i$ until the final paragraph. By Lemma~\ref{lm:basic}(iii) there is an $F$-null set $N\subseteq V$ (the types that need not best-respond) such that every $v\in V\setminus N$ best-responds and the bid supports $R(v):=\operatorname{supp}\sigma(v)$ are ordered by type.

\emph{Proof of (ii) Regularity:} by definition,
    \[
        b(p)\coloneq G^{-1}(p)=\inf\{b\ge 0:G(b)\ge p\}
    \]
is weakly increasing. It is discontinuous at $p_0\in(0,1)$ only if $G$ is flat at level $p_0$, i.e.\ $G(a)=G(\tilde a)=p_0$ for some $a<\tilde a$. By the no-gap property in Lemma~\ref{lm:basic}(i), $G$ has no flat at any level in $(0,1)$, so $b$ is continuous on $(0,1]$; together with $b(p)\to0$ as $p\to0^+$ (the bid support has infimum zero by the common-support property, Lemma~\ref{lm:basic}(iv)), this gives continuity on $[0,1]$. Likewise $b$ is constant on an interval of quantiles only if $G$ has an atom at the corresponding bid. By the no-atom property in Lemma~\ref{lm:basic}(ii), $G$ has no atom at any positive bid, so $b$ is strictly increasing wherever $b(p)>0$. The single admissible atom of $G$ is at bid zero, of size $G(0)$, producing the initial flat $\{p:b(p)=0\}=[0,G(0)]$.

\emph{Proof of (i) Equivalence:} write $\mu$ for the joint distribution of (value, bid) in the BNE of the original auction. Both $\mu$ and the distribution of $(Q(p),b(p))$ for $p\sim U[0,1]$ have value-marginal $F$. Hence, to show that the two conditional bid distributions given the value agree for $F$-almost every $v$, it suffices to show that the two joint distributions coincide.

\begin{claim}\label{claim:frechet}
For all $(x,y),\quad\mu\big((-\infty,x]\times(-\infty,y]\big)=\min\{F(x),G(y)\}.$%
\footnote{This is the Fr\'echet--Hoeffding upper bound for the joint c.d.f.\ with marginals $F$ and $G$ (see, e.g., \citealt{nelsen2006introduction}, \S2.5); the proof here is self-contained and requires no further reference.}
\end{claim}

\begin{proof}
Write $H(x,y):=\mu((-\infty,x]\times(-\infty,y])$ and set $B^1:=(-\infty,x]\times(y,\infty)$, $B^2:=(x,\infty)\times(-\infty,y]$. Because the marginals of $\mu$ are $F$ and $G$,
\[
    \mu(B^1)=F(x)-H(x,y),\qquad \mu(B^2)=G(y)-H(x,y),
\]
so $H(x,y)=\min\{F(x),G(y)\}$ once we show $\min\{\mu(B^1),\mu(B^2)\}=0$. Suppose instead that $\mu(B^1)>0$ and $\mu(B^2)>0$. Since
\[
    \mu(B^1)=\int_{\{v\le x\}}\sigma(v)\big((y,\infty)\big)\,\mathrm{d}F(v)>0,
\]
and $N$ is $F$-null, there is $v^1\le x$ with $v^1\notin N$ and $\sigma(v^1)((y,\infty))>0$. Hence there exists $b^1\in R(v^1)$ satisfying $b^1>y$. Symmetrically, $\mu(B^2)>0$ yields $v^2>x$ with $v^2\notin N$ and some $b^2\in R(v^2)$ satisfying $b^2\le y$. But then $v^1<v^2$ both outside $N$, while $b^1>y\ge b^2$, contradicting the monotonicity in Lemma~\ref{lm:basic}(iii).
\end{proof}

On the other hand
\[
    \Pr\big(Q(p)\le x,\;b(p)\le y\big)
    =\Pr\big(p\le F(x),\;p\le G(y)\big)=\min\{F(x),G(y)\},
\]
so $(Q(p),b(p))$ has the same joint c.d.f.\ as $\mu$. A joint c.d.f.\ determines the distribution on $\mathbb{R}^2$, hence the two joint distributions coincide.

Finally we prove that $(b_1,b_2)$ is a BNE of the quantile auction. Restore the index $i$ and fix a bidder $i$. In the quantile auction, by \eqref{eq: lemma increasing bid}, the opponent's unconditional bid c.d.f.\ induced by $b_{-i}$ is $G_{-i}$, which is the same as in the original auction. Since $BR_i(v)$ depends only on $v$ and the opponent's bid distribution, it is the same in both auctions. By part (i), the conditional distribution of $b_i(p_i)$ given $Q_i(p_i)=v$ equals $\sigma_i(v)$ for $F_i$-a.e.\ $v$, and the original BNE gives $\sigma_i(v)(BR_i(v))=1$ for $F_i$-a.e.\ $v$. Hence $b_i(p_i)\in BR_i(Q_i(p_i))$ for a.e.\ $p_i$. As this holds for $i=1,2$, the profile $(b_1,b_2)$ is a BNE of the quantile auction.

\subsection{Proof of Theorem~\ref{main theo}}\label{appendix:proof main theo}
Let $(b_1,b_2)$ be a monotone quantile representative of an equilibrium satisfying the regularity condition in Lemma~\ref{increasing bid function}(ii). By Lemma~\ref{increasing bid function}, any equilibrium has such a monotone representation, so the restriction is without loss of generality. For each bidder $i$, let $P_i$ be the set of quantiles at which bidder $i$'s equilibrium bid is a best response:
\[
P_i:=\{p\in[0,1]: b_i(p)\in BR_i(Q_i(p))\}.
\]
By the definition of equilibrium in the quantile auction, $\lambda(P_i)=1$, where $\lambda$ denotes Lebesgue measure. In what follows, optimality of the equilibrium bid will only be used for quantiles in $P_i$. Since $P_i$ has full measure, it is dense in every non-degenerate interval.

By Lemma~\ref{lm:basic}, at most one bidder can have a strictly positive probability of bidding zero in the equilibrium. We first consider the case in which bidder 2 has no positive mass at bid zero. After deriving \eqref{eq:S(phi)}, we show that this case is the only one consistent with the labeling convention \eqref{relabel}.

Let $\alpha=\sup\{p:b_1(p)=0\}$. By the no-gap property and monotonicity in Lemma~\ref{lm:basic}, bidder 1 bids zero if and only if $p\in[0,\alpha]$. By Lemma~\ref{increasing bid function}, since bidder 2 has no positive mass at bid zero, $b_2$ is continuous and strictly increasing on $[0,1]$, with $b_2(0)=0$ and $b_2(1)=\bar b$, where $\bar b$ is the common upper endpoint of the bid supports.

Hence $b_2$ is a bijection from $[0,1]$ to $[0,\bar b]$, so for every $p\in[\alpha,1]$ there is a unique quantile $\phi(p)\in[0,1]$ such that
\begin{equation}\label{eq:def phi}
    b_1(p)=b_2(\phi(p)),
\end{equation}
that is, $\phi=b_2^{-1}\circ b_1$ on $[\alpha,1]$. As a composition of continuous, strictly increasing functions, $\phi$ is continuous and strictly increasing on $[\alpha,1]$. Moreover,
\[
\phi(\alpha)=0,
\qquad
\phi(1)=1.
\]
Thus the restriction
\[
\phi|_{[\alpha,1]}:[\alpha,1]\to[0,1]
\]
has a continuous inverse. Throughout the proof, $\phi^{-1}$ denotes this inverse, so in particular $\phi^{-1}(0)=\alpha$.

For notational convenience, extend $\phi$ to all of $[0,1]$ by setting
\[
\phi(p)=0,\qquad p\in[0,\alpha).
\]
This extension is consistent with the quantile matching interpretation: bidder 1 bids zero for all $p\in[0,\alpha]$, while bidder 2 bids zero only at quantile $0$.

\begin{lemma}
\label{Lem:conv}
$b_1,b_2$ are convex, hence absolutely continuous and almost everywhere differentiable.
\end{lemma}
\begin{proof}
The maximization problem of bidder 1 is given by\footnote{Since bidder 2 has no positive mass at bid zero, the expression of bidder 1's problem is valid for all $b$. For the case of bidder 2 which we use later, since we restrict $q\in (0,1]$, hence $\phi^{-1}(q)\in (\alpha,1]$, the expression is also valid as the bid is always strictly positive.}
\begin{equation}
    \label{maximization value}
    \max_b vG_2(b)-b.
\end{equation}
Equivalently, bidder 1's strategy can be viewed as choosing a quantile $q$ of bidder 2 to match:
\begin{equation}
    \label{maximization}
    \max_{q\in [0,1]} Q_1(p)q-b_2(q)= \max_{\tilde p\in [0,1]} Q_1(p)\phi(\tilde p)-b_1(\tilde p),
\end{equation}
where the equality comes from \eqref{eq:def phi}. For any $p \in P_1\cap (\alpha,1]$ and $\tilde p\in (\alpha,1]$, since bidder 1 with quantile $p$ prefers the equilibrium bid $b_1(p)=b_2(\phi(p))$ to any deviation $b_1(\tilde p)=b_2(\phi(\tilde p))$, we have
\begin{equation}\label{eq:optimality}
    Q_1(p)\phi(p)-b_1(p)\geq Q_1(p)\phi(\tilde p)-b_1(\tilde p).
\end{equation}
Let
\[
q=\phi(p),\quad \tilde q=\phi(\tilde p),
\]
equation~\eqref{eq:optimality} becomes
\[
Q_1(\phi^{-1}(q))q-b_2(q)\geq Q_1(\phi^{-1}(q))\tilde q-b_2(\tilde q),
\]
which is equivalent to
\begin{equation}\label{eq:supporting hyperplane}
     b_2(\tilde q)\ge b_2(q)+Q_1(\phi^{-1}(q))(\tilde q-q)
\end{equation}
for all $\tilde q\in \phi((\alpha,1])=(0,1]$ and all $q\in \phi(P_1\cap (\alpha,1])$. Since $\phi$ is a homeomorphism between $(\alpha,1]$ and $(0,1]$, and $P_1\cap (\alpha,1]$ is dense in $(\alpha,1]$, $\phi(P_1\cap (\alpha,1])$ is dense in $(0,1]$. Equation \eqref{eq:supporting hyperplane} implies a supporting hyperplane at every point of the dense set $\phi(P_1\cap(\alpha,1])$, which suffices to establish the convexity by the following lemma.
\begin{lemma}\label{lm:convex}
Let $f:I\to\mathbb R$ be continuous. Suppose there is a dense set $D\subset I$ such that for every $y\in D$, there exists $m_y\in\mathbb R$ satisfying
\[
f(x)\ge f(y)+m_y(x-y)\qquad\forall x\in I.
\]
Then $f$ is convex.
\end{lemma}

\begin{proof}
Consider any $x<z\in I$. For any $t\in (0,1)$ and $y=t x+(1-t)z$, pick $y_n\in D,y_n\to y$, apply the inequality with $(x,y_n)$ and $(z,y_n)$ to obtain
    \[
    f(x)\ge f(y_n)+m_n(x-y_n),
\qquad
f(z)\ge f(y_n)+m_n(z-y_n).
\]
Multiply the first inequality by $(z-y_n)/(z-x)$ and the second by $(y_n-x)/(z-x)$. The coefficients are nonnegative and sum to one, and the terms involving $m_n$ cancel. Hence
\[
\frac{z-y_n}{z-x}f(x)+\frac{y_n-x}{z-x}f(z)\ge f(y_n).
\]
Taking $n\to\infty$ and using continuity of $f$, we obtain
\[
tf(x)+(1-t) f(z)\ge f(y).
\]
Thus $f$ is convex.
\end{proof}

By Lemma~\ref{lm:convex}, $b_2$ is convex on $(0,1]$. Continuity extends convexity to $[0,1]$. A symmetric argument implies
\begin{equation}\label{eq:b1}
    b_1(\tilde p)\ge b_1(p)+Q_2(\phi(p))(\tilde p-p)
\end{equation}
for all $p\in \phi^{-1}(P_2\cap (0,1])$ and $\tilde p\in (\alpha,1]$, which proves that $b_1$ is convex on $(\alpha,1]$.\footnote{The restriction of $\phi$ is a homeomorphism from $(\alpha,1]$ to $(0,1]$, and therefore $\phi^{-1}$ is a homeomorphism from $(0,1]$ to $(\alpha,1]$.} Since we also have $b_1(p)=b_1(\alpha)=0$ for all $p\in [0,\alpha]$ and that $b_1$ is weakly increasing on $[\alpha,1]$, the convexity can be extended to $[0,1]$. It then implies that $b_1,b_2$ are absolutely continuous, hence differentiable a.e.
\end{proof}

\begin{lemma}
\label{lem:abso-conti}
The function $\phi$ is locally absolutely continuous on $(\alpha,1]$, i.e., for any $\delta>0$, the function $\phi$ is absolutely continuous on $[\alpha+\delta,1]$.
\end{lemma}
\begin{proof}
We prove $\phi$ is actually locally Lipschitz. In \eqref{eq:b1}, fix $x<y$ in $[\alpha,1]$, choose $y_n\in \phi^{-1}(P_2\cap(0,1])$ with $y_n\to y$ and $y_n>x$ for all sufficiently large $n$. Apply \eqref{eq:b1}\footnote{By continuity, this also holds for $x=\alpha$.} with $(p,\tilde p)=(y_n,x)$ to obtain
\[
b_1(y_n)-b_1(x)\le Q_2(\phi(y_n))(y_n-x)\le Q_2(1)(y_n-x).
\]
Taking $n\to\infty$ and using continuity of $b_1$, we get
\[
b_1(y)-b_1(x)\le Q_2(1)(y-x).
\]
Thus $b_1$ is Lipschitz with $L=Q_2(1)< \infty$. Also, for any $\delta>0$ and $x<y\in [\phi(\alpha+\delta),1]$ (note that $\phi(\alpha+\delta)>0$ and $Q_1(\alpha+\delta)>0$), similarly by \eqref{eq:supporting hyperplane} we obtain
\[
b_2(y)-b_2(x)\ge Q_1(\alpha+\delta)(y-x).
\]
So, $b_2^{-1}$ is Lipschitz on $[b_2(\phi(\alpha+\delta)),\bar{b}]$ with $L=\frac{1}{Q_1(\alpha+\delta)}<\infty$.

Since $\phi=b_2^{-1}\circ b_1$, it is Lipschitz with $L=\frac{Q_2(1)}{Q_1(\alpha+\delta)}$, hence absolutely continuous on $[\alpha+\delta,1]$.
\end{proof}
By the symmetric argument as Lemma~\ref{lem:abso-conti}, $\phi^{-1}$ is locally absolutely continuous on $(0,1]$. With these regularity properties of $b_i$ and $\phi$ in hand, we can differentiate the relevant identities a.e. Recall bidder 1's problem \eqref{maximization}
\begin{equation*}
    \max_{\tilde p\in [0,1]} Q_1(p)\phi(\tilde p)-b_1(\tilde p)=\max_{\tilde p\in [0,1]} Q_1(p)\phi(\tilde p)-b_2(\phi(\tilde p)).
\end{equation*}
For $p\in P_1\cap(\alpha,1]$, the choice $\tilde p=p$ is optimal. Since $b_2$ is convex, the optimality condition is equivalent to
\begin{equation}
\label{subdifferential 1}
    Q_1(p)\in \partial b_2(\phi(p)) \qquad\text{for }p\in P_1\cap(\alpha,1].
\end{equation}
A symmetric argument applied to bidder 2's optimality gives
\begin{equation}
\label{subdifferential 2}
    Q_2(\phi(p))\in \partial b_1(p) \qquad\text{for }\phi(p)\in P_2\cap(0,1].
\end{equation}
Let $D_i$ be the set where $b_i$ is differentiable. Since $b_i$ is convex, $\lambda(D_i^c)=0$. Since $\phi^{-1}$ is locally absolutely continuous, it maps Lebesgue-null sets to Lebesgue-null sets. Hence, on each interval $[\alpha+\delta,1]$,
\[
P_1^c,\quad \phi^{-1}(D_2^c),\quad D_1^c,\quad \phi^{-1}(P_2^c)
\]
are null sets. Therefore, for a.e. $p\in[\alpha+\delta,1]$,
\begin{equation}
\label{foc}
    Q_1(p)=b_2'(\phi(p)),
\qquad
Q_2(\phi(p))=b_1'(p).
\end{equation}

We introduce a useful lemma here, which is a standard result.
\begin{lemma}
\label{chain rule}
If $f$ is absolutely continuous and $g$ is absolutely continuous and monotone, then $f\circ g$ is absolutely continuous, hence differentiable a.e. Moreover, $(f\circ g)'=(f'\circ g)g'$ a.e.
\end{lemma}
\begin{proof}
Without loss, assume $g$ is increasing. Monotonicity of $g$ ensures that whenever $(a_1,b_1),\ldots,(a_n,b_n)$ are disjoint, the intervals $(g(a_1),g(b_1)),\ldots,(g(a_n),g(b_n))$ are disjoint. Absolute continuity of $f\circ g$ then comes from the definition. The proof of the chain rule can be found in \citet{leoni2017first} Corollary 3.50.
\end{proof}

Apply Lemma~\ref{chain rule} to differentiate $b_1(p)=b_2(\phi(p))$. Since $\phi$ is increasing, the chain rule is valid a.e.\ on $[\alpha+\delta,1]$:
\[
b_1'(p)=\frac{d}{dp}b_2(\phi(p))=b_2'(\phi(p))\phi'(p).
\]
Combine it with \eqref{foc} to obtain
\begin{equation}
\label{ode}
    Q_2(\phi(p))=Q_1(p)\phi'(p) \quad \text{a.e. } p\in [\alpha+\delta,1].
\end{equation}
Rewrite \eqref{ode} as
\begin{equation}
\label{ode rearrange}
    \frac{\phi'(p)}{Q_2(\phi(p))}=\frac{1}{Q_1(p)}, \quad \text{a.e. } p\in [\alpha+\delta,1].
\end{equation}
Recall the definition of $S_i(p)$ from \eqref{def S}. The function $S_1$ is absolutely continuous on $[\alpha+\delta,1]$ because $\frac{1}{Q_1}$ is bounded on $[\alpha+\delta,1]$. Similarly, $S_2$ is absolutely continuous on $[\phi(\alpha+\delta),1]$. Hence, by Lemma~\ref{chain rule}, $S_2(\phi(p))$ is also absolutely continuous. Notice the LHS of \eqref{ode rearrange} equals $-(S_2(\phi(p)))'$ a.e., while the RHS equals $-S'_1(p)$ a.e. So,
\begin{equation}
\label{eq:local S(phi(p))}
    (S_2(\phi(p)))'=S'_1(p) \quad \text{a.e. } p\in [\alpha+\delta,1].
\end{equation}
Since \eqref{eq:local S(phi(p))} holds for any $\delta>0$, it follows that
\[
S_2(\phi(p))-S_1(p)= \text{const on } (\alpha,1].
\]
At $p=1$, $\phi(p)=1$ (by the common bid support property in Lemma~\ref{lm:basic}), $S_2(\phi(1))-S_1(1)=0-0=0$. Thus,
\begin{equation}
    \label{eq:S(phi)}
    S_2(\phi(p))=S_1(p)  \quad \forall p \in (\alpha,1].
\end{equation}
Recall that $\phi(\alpha)=0$ and that $\phi$ is continuous on $[\alpha,1]$. Take limit $p\to \alpha^+$ in \eqref{eq:S(phi)}. If $\alpha>0$, then the RHS has a finite limit $S_1(\alpha)$, which implies the LHS has a finite limit and equals the RHS, i.e. $S_2(0)=S_1(\alpha)<S_1(0)$, which is consistent with the labeling convention \eqref{relabel}. In this case, we can write $\alpha=S_1^{-1}(S_2(0))$. If $\alpha=0$, then taking $p\to0^+$ in \eqref{eq:S(phi)} gives one of two cases: either both limits are finite and $S_2(0)=S_1(0)$, or both limits diverge to $+\infty$. In either case, the labeling convention \eqref{relabel} is satisfied.

If instead bidder 2 had a strictly positive mass $\hat{\alpha}>0$ at bid zero, then the symmetric version of the above argument would imply
\[
S_1(0)=S_2(\hat \alpha).
\]
This contradicts $S_2(\hat \alpha)<S_2(0)\leq S_1(0)$, where the last inequality comes from \eqref{relabel}.

Since $S_2$ is continuous and strictly decreasing, $\phi$ is uniquely pinned down on $(\alpha,1]$ by
\[
\phi(p)=S_2^{-1}(S_1(p)).
\]
With $\phi$ determined, integrating \eqref{foc} recovers the bid function:
\[
b_1(p) = 
\begin{cases}
    0 & p\in[0,\alpha]\\
    \int_{\alpha}^p Q_2(\phi(s))ds &p\in (\alpha,1],
\end{cases}
\qquad
b_2(p) = \int_0^p Q_1(\phi^{-1}(s)) \, ds.
\]
To verify this is indeed an equilibrium, first note that since $Q_i$ and $\phi$ are increasing, $b_1$ and $b_2$ are convex. By construction, the left and right derivatives of $b_2$ at $\phi(p)$ are the left and right limits of $Q_1(\phi^{-1}(\cdot))$ at $\phi(p)$, which bracket $Q_1(p)$. Similarly, the left and right derivatives of $b_1$ at $p$ are the left and right limits of $Q_2(\phi(\cdot))$ at $p$, which bracket $Q_2(\phi(p))$. Hence the subdifferential optimality conditions \eqref{subdifferential 1} and \eqref{subdifferential 2} are satisfied for all (instead of a.e.\ ) $p\in(\alpha,1]$ for bidder 1 and for all $\phi(p)\in(0,1]$ for bidder 2.

For bidder 1 with quantile $p=\alpha$ who is going to bid 0, deviating to a positive bid corresponds to choosing some bidder 2 quantile $p>0$. Since $b_2$ is convex, it is enough to check the right derivative at zero:
\[
\lim_{p\to 0^+}b_2'(p)=\lim_{p\to 0^+} Q_1(\phi^{-1}(p))\geq Q_1(\alpha).
\]
So it is optimal for bidder 1 with $p=\alpha$ and value $Q_1(\alpha)$ to bid 0. Since $Q_1(p)\le Q_1(\alpha)$ for every $p<\alpha$, bidding zero is also optimal for bidder 1 at every such quantile.

At $p=0$, bidder 2 may fail to have an attained best response when bidder 1 has a zero-bid mass. This does not affect equilibrium under the present definition, because $\{0\}$ is a Lebesgue-null set of quantiles. Hence, $(b_1,b_2)$ is a BNE.

\subsection{Proof of Corollary~\ref{coro:efficient}}\label{Appendix:coro efficient}
Sufficiency is standard. For necessity, by Theorem~\ref{main theo} bidder 1 with quantile $p_1>\alpha$ wins with probability one when bidder 2's quantile satisfies $p_2<\phi(p_1)$, and loses with probability one when $p_2>\phi(p_1)$. Efficiency therefore requires
    \begin{equation}\label{eq:efficient}
        Q_2(\phi(p)) = Q_1(p)
    \end{equation}
except possibly at the (at most countably many) jump points of $Q_2$. Combining this with the equilibrium identity \eqref{ode} gives
    \[
        Q_2(\phi(p)) = Q_1(p)\phi'(p) \quad \text{a.e.\ } p\in(\alpha,1],
    \]
and substituting \eqref{eq:efficient} yields $\phi'(p)=1$ a.e. Since $\phi$ is absolutely continuous on $(\alpha,1]$ by Lemma~\ref{lem:abso-conti}, integration gives $\phi(p)=p+C$ for some constant $C$, and the boundary condition $\phi(1)=1$ forces $C=0$. Hence $\phi(p)=p$ on $(\alpha,1]$; since $\phi$ maps $(\alpha,1]$ onto $(0,1]$, this requires $\alpha=0$.

By Theorem~\ref{main theo}, $\phi=S_2^{-1}\circ S_1$, so $\phi=\mathrm{id}$ on $(0,1]$ yields $S_1(p)=S_2(p)$ for all $p\in[0,1]$. Differentiating gives $Q_1(p)=Q_2(p)$ a.e., and left continuity of $Q_i$ promotes this to equality everywhere on $(0,1]$. Therefore $F_1=F_2$.

\subsection{Proof of Theorem~\ref{theo:bid fosd}}\label{Appendix: bid fosd}
By definition of the quantile function $Q_i$, we have $\hat F_i\succeq_{FOSD}  F_i$ if and only if $\hat Q_i(p)\geq Q_i(p),\forall p\in (0,1)$.

Recall the equilibrium bid function \eqref{bid function}. First consider the case that $i=1$. To show the FOSD relationship of bid distributions, it is equivalent to show $\hat b_1(p)\geq b_1(p)$ for all $p\in [0,1]$, as the bid function is weakly increasing in $p$.

It is helpful to write the $\phi$ function in \eqref{bid function} more explicitly:
    \[
     b_1(p) = 
    \begin{cases}
        0 & p\in[0,\alpha]\\
        \displaystyle\int_{\alpha}^p Q_2(S_2^{-1}(S_1(s)))\,ds & p\in (\alpha,1],
    \end{cases}
    \]
where
    \[
    \alpha =
    \begin{cases}
    0 & \text{if } S_1(0)=S_2(0) \\
    S_1^{-1}(S_2(0)) & \text{if } S_1(0)>S_2(0).
    \end{cases}
    \]
Define $K_i(s)\coloneq Q_i(S_i^{-1}(s))$ for all $s\in [0,S_i(0)]$. Note $K_2$ is only a function of $Q_2$, which is fixed when we change $Q_1$. The function $S_2$ is strictly decreasing, and so is $S_2^{-1}$. Also, $Q_2$ is weakly increasing, hence $K_2$ is weakly decreasing. Since $\hat Q_1(p)\geq Q_1(p),\forall p\in (0,1)$, we have $\hat S_1(p)\leq S_1(p), \forall p\in [0,1]$. Combining the facts, the integrand satisfies
    \[
    Q_2(S_2^{-1}(S_1(s)))=K_2(S_1(s))\leq K_2(\hat S_1(s)),\quad \forall s\in(\alpha,1],
    \]
so it increases weakly when $Q_1$ changes to $\hat Q_1$. Next consider the change in $\alpha$. Since $\alpha$ is the solution to $S_1(\alpha)=S_2(0)$, and $S_i$ is strictly decreasing, we have $\hat \alpha\leq \alpha$, where $\hat \alpha$ satisfies $\hat S_1 (\hat \alpha)=S_2(0)$. So, both the integrand and the range of integration become weakly larger, hence the integral weakly increases, i.e. $\hat b_1(p)\geq b_1(p)$.

For the case of $i=2$, the logic is the same, as $b_2(p)=\int_0^p K_1(S_2(s))ds$ and the range does not change.

If instead the shift from $Q_i$ to $\hat Q_i$ reverses the labeling, i.e. $S_i(0)>S_j(0)>\hat S_i(0)$, so that bidder $i$ is labeled bidder 1 under $Q_i$ and bidder 2 under $\hat Q_i$, then by the same logic the integrand changes from $K_j(S_i(s))$ to $K_j(\hat S_i(s))$, and the range changes from $(\alpha,p)$ to $(0,p)$. So the integral weakly increases, hence $\hat b_i(p)\geq b_i(p)$.

\subsection{Proof of Lemma~\ref{lem:revenue derivatives}}\label{proof: marginal revenue}
\begin{proof}
Fix $\epsilon_0\in(0,\bar\epsilon)$. Since $Q-\epsilon_0H\geq 0$, we have $H\leq C Q$ for $C=1/\epsilon_0$. This bound will be used in the dominated-convergence arguments below.

First consider the right derivative, i.e.\ locally strengthening one bidder. Let the perturbed bidder have quantile function $Q_\epsilon=Q+\epsilon H$, and let
    \[
        S(p)=\int_p^1\frac{dt}{Q(t)},
        \qquad
        S_\epsilon(p)=\int_p^1\frac{dt}{Q_\epsilon(t)}.
    \]
Also define
    \[
        T_\epsilon(p):=\frac{S(p)-S_\epsilon(p)}{\epsilon}
        =\frac{1}{\epsilon}\left(\int_p^1\frac{dt}{Q(t)}-\int_p^1\frac{dt}{Q_\epsilon(t)}\right)=
        \int_p^1\frac{H(t)}{Q(t)Q_\epsilon(t)}\,dt.
    \]
Then $T_\epsilon(p)\to T_H(p)$ for every $p\in(0,1]$.

For $\epsilon>0$, the perturbed bidder is the stronger bidder. We label this bidder as bidder 2, which is consistent with the labeling \eqref{relabel}. Denote by $\phi_\epsilon$ the equilibrium quantile matching function. Let $\eta_\epsilon=\phi_\epsilon^{-1}$. Then $\eta_\epsilon(p)\geq p$ for all $p$ and
    \[
        S(\eta_\epsilon(p))=S_\epsilon(p).
    \]
By \eqref{eq:revenue}, we obtain
    \[
        R_H(\epsilon)
        =
        \int_0^1
        (2-p-\eta_\epsilon(p))Q(\eta_\epsilon(p))\,dp.
    \]
At $\epsilon=0$, we have $\eta_0(p)=p$, hence
    \[
        R_H(0)=\int_0^1 2(1-p)Q(p)\,dp.
    \]
Therefore,
    \[
        \frac{R_H(\epsilon)-R_H(0)}{\epsilon}
        =
        A_\epsilon-B_\epsilon,
    \]
where
    \[
        A_\epsilon
        =
        \frac{2}{\epsilon}
        \int_0^1(1-p)
        \left[Q(\eta_\epsilon(p))-Q(p)\right]dp,
    \]
and
    \[
        B_\epsilon
        =
        \frac{1}{\epsilon}
        \int_0^1
        (\eta_\epsilon(p)-p)Q(\eta_\epsilon(p))\,dp.
    \]
We compute the limits of $A_\epsilon$ and $B_\epsilon$ separately.

    \begin{claim}\label{claim:general-right-B}
        \[
            B_\epsilon
            \xrightarrow{\epsilon\to0^+}
            \int_0^1T_H(p)Q(p)^2\,dp.
        \]
    \end{claim}

    \begin{proof}
Since $S(\eta_\epsilon(p))=S_\epsilon(p)$, we have
        \[
            S(p)-S(\eta_\epsilon(p))=S(p)-S_\epsilon(p)
            =\epsilon T_\epsilon(p).
        \]
Equivalently,
        \[
            \int_p^{\eta_\epsilon(p)}\frac{dt}{Q(t)}
            =
            \epsilon T_\epsilon(p).
        \]
Since $\frac{1}{Q}$ is weakly decreasing,
        \[
            \frac{\eta_\epsilon(p)-p}{Q(\eta_\epsilon(p))}
            \leq
            \int_p^{\eta_\epsilon(p)}\frac{dt}{Q(t)}
            \leq
            \frac{\eta_\epsilon(p)-p}{Q(p)}.
        \]
Rearranging gives
        \[
            T_\epsilon(p)Q(p)Q(\eta_\epsilon(p))
            \leq
            \frac{(\eta_\epsilon(p)-p)Q(\eta_\epsilon(p))}{\epsilon}
            \leq
            T_\epsilon(p)Q(\eta_\epsilon(p))^2.
        \]
As $\epsilon\to0^+$, $\eta_\epsilon(p)\to p^+$. Since $Q$ has at most countably many discontinuity points,
        \[
            Q(\eta_\epsilon(p))\to Q(p)
            \qquad \text{for a.e. }p.
        \]
Hence
        \[
            \frac{(\eta_\epsilon(p)-p)Q(\eta_\epsilon(p))}{\epsilon}
            \to
            T_H(p)Q(p)^2
            \qquad \text{for a.e. }p.
        \]

It remains to justify the use of Lebesgue's dominated convergence theorem. Since $H\leq CQ$, we have
        \[
            T_\epsilon(p)\leq C S(p).
        \]
Also,
        \[
            S_\epsilon(p)
            =
            \int_p^1\frac{dt}{Q(t)+\epsilon H(t)}
            \geq
            \frac{1}{1+C\epsilon}S(p).
        \]
Because $S(\eta_\epsilon(p))=S_\epsilon(p)$, this implies
        \[
            S(p)\leq(1+C\epsilon)S(\eta_\epsilon(p)).
        \]
Therefore, for $\epsilon$ small enough,
        \[
        \begin{aligned}
             \frac{(\eta_\epsilon(p)-p)Q(\eta_\epsilon(p))}{\epsilon} &\leq 
            T_\epsilon(p)Q(\eta_\epsilon(p))^2
            \leq
            C(1+C\epsilon)S(\eta_\epsilon(p))Q(\eta_\epsilon(p))^2\\
           & \leq
            C(1+C\epsilon)(1-\eta_\epsilon(p))Q(\eta_\epsilon(p))
            \leq
            2C Q(1).
        \end{aligned}
        \]
The integrand is uniformly bounded, hence the dominated convergence theorem proves the claim.
    \end{proof}

    \begin{claim}\label{claim:general-right-A}
        \[
            A_\epsilon
            \xrightarrow{\epsilon\to0^+}
            2\int_{(0,1)}(1-p)T_H(p)Q(p)\,dQ(p).
        \]
    \end{claim}

    \begin{proof}
Let $\mu=dQ$ be the Stieltjes measure generated by the non-decreasing, left-continuous function $Q$, with
        \[
           \mu([p_1,p_2))=Q(p_2)-Q(p_1), \quad 0\le p_1\le p_2\le1.
        \]
By Fubini's theorem,
        \[
        \begin{aligned}
            \int_0^1(1-p)
            \left[Q(\eta_\epsilon(p))-Q(p)\right]dp
            &=
            \int_0^1(1-p)\mu([p,\eta_\epsilon(p)))\,dp \\
            &=
            \int_{t\in (0,1)}dQ(t)
            \int_0^1(1-p)
            \mathbf 1\{p\leq t<\eta_\epsilon(p)\}\,dp.
        \end{aligned}
        \]
We can replace $\int_0^1$ by $\int_{t\in (0,1)}$ because they are equal as the inner integral is always zero when $t=0$ or $t=1$ (because $\eta_\epsilon(p)\leq 1$ by definition), and the current format makes it easier to deal with the boundary of the Lebesgue-Stieltjes integral. The condition $t<\eta_\epsilon(p)$ is equivalent to
        \[
            S(t)>S(\eta_\epsilon(p))=S_\epsilon(p).
        \]
Define
        \[
            p_\epsilon(t)=
            \begin{cases}
                S_\epsilon^{-1}(S(t)) & \text{if } S(t)<S_\epsilon(0),\\
                0 & \text{if } S(t)\geq S_\epsilon(0).
            \end{cases}
        \]
Then
        \[
            p\leq t<\eta_\epsilon(p)
            \iff
            p\in(p_\epsilon(t),t],
        \]
up to endpoints that do not affect the inner Lebesgue integral. Hence
        \[
            A_\epsilon
            =
            2\int_{(0,1)}
            \frac{1}{\epsilon}
            \left(
                \int_{p_\epsilon(t)}^t(1-p)\,dp
            \right)dQ(t).
        \]

For each fixed $t\in(0,1)$, when $\epsilon$ is small enough, $p_\epsilon(t)>0$ and
        \[
            S_\epsilon(p_\epsilon(t))=S(t).
        \]
Thus
        \[
            \int_{p_\epsilon(t)}^t\frac{du}{Q_\epsilon(u)}
            = S_\epsilon(p_\epsilon(t))-S_\epsilon(t)
            =S(t)-S_\epsilon(t)
            =
            \epsilon T_\epsilon(t).
        \]
Since $p_\epsilon(t)\to t^-$ as $\epsilon\to 0^+$ for all $t\in (0,1)$, the local average of $1/Q_\epsilon$ over $(p_\epsilon(t),t)$ converges to $1/Q(t)$ by left-continuity of $Q$. Hence
        \[
            \frac{t-p_\epsilon(t)}{\epsilon}
            \to
            T_H(t)Q(t).
        \]
Therefore,
        \[
            \frac{1}{\epsilon}
            \int_{p_\epsilon(t)}^t(1-p)\,dp
            \to
            (1-t)T_H(t)Q(t).
        \]

We again use the dominated convergence theorem. Since
        \[
            \int_{p_\epsilon(t)}^t\frac{du}{Q_\epsilon(u)}
            =
            \epsilon T_\epsilon(t),
        \]
monotonicity of $Q_\epsilon$ gives
        \[
            \frac{t-p_\epsilon(t)}{Q_\epsilon(t)}
            \leq
            \epsilon T_\epsilon(t).
        \]
Hence, for $\epsilon$ small enough,
        \[
            \frac{1}{\epsilon}
            \int_{p_\epsilon(t)}^t(1-p)\,dp
            \leq
            \frac{t-p_\epsilon(t)}{\epsilon}
            \leq
            T_\epsilon(t)Q_\epsilon(t)
            \leq
            C(1+C\epsilon)S(t)Q(t)
            \leq
            2C.
        \]
Since $dQ$ is a finite measure on $(0,1)$, dominated convergence gives
        \begin{equation*}
            A_\epsilon
            \to
            2\int_{(0,1)}(1-t)T_H(t)Q(t)\,dQ(t).
            \qedhere
        \end{equation*}
    \end{proof}

Combining Claims~\ref{claim:general-right-B} and \ref{claim:general-right-A}, we obtain
    \begin{equation}\label{eq:general-right-raw}
        (R_H)'_+(0)
        =
        2\int_{(0,1)}(1-p)T_H(p)Q(p)\,dQ(p)
        -
        \int_0^1T_H(p)Q(p)^2\,dp.
    \end{equation}

    \begin{claim}\label{claim:general-right-stieltjes}
        \[
            (R_H)'_+(0)=I_H-B_H-J_H.
        \]
    \end{claim}

    \begin{proof}
Let
        \[
            M(p)=(1-p)Q(p)^2.
        \]
The Stieltjes differential of $Q^2$ is not always $2Q\,dQ$. At a jump point $p$,
        \[
            \Delta(Q(p)^2)
            =
            Q(p^+)^2-Q(p)^2
            =
            2Q(p)\Delta Q(p)+(\Delta Q(p))^2.
        \]
Hence
        \[
        \begin{aligned}
            \int_{(0,1)}T_H(p)\,dM(p)
            &=
            -\int_0^1T_H(p)Q(p)^2\,dp \\
            &\quad
            +2\int_{(0,1)}(1-p)T_H(p)Q(p)\,dQ(p) \\
            &\quad
            +\sum_{p\in(0,1)}
            (1-p)T_H(p)(\Delta Q(p))^2.
        \end{aligned}
        \]
Comparing this identity with \eqref{eq:general-right-raw} gives
        \[
            (R_H)'_+(0)
            =
            \int_{(0,1)}T_H(p)\,dM(p)-J_H.
        \]
Since
        \[
            dT_H(p)=-\frac{H(p)}{Q(p)^2}\,dp,
        \]
integration by parts yields
        \[
        \begin{aligned}
            \int_{(0,1)}T_H(p)\,dM(p)
            &=
            \lim_{p\to1^-}T_H(p)(1-p)Q(p)^2
            -
            \lim_{p\to0^+}T_H(p)(1-p)Q(p)^2 \\
            &\quad
            -
            \int_0^1(1-p)Q(p)^2\,dT_H(p) \\
            &=
            -B_H+\int_0^1(1-p)H(p)\,dp \\
            &=
            -B_H+I_H.
        \end{aligned}
        \]
Therefore,
        \begin{equation*}
            (R_H)'_+(0)=I_H-B_H-J_H.
            \qedhere
        \end{equation*}
    \end{proof}

We next consider the left derivative. Let $\epsilon>0$, and define
    \[
        \hat Q_\epsilon=Q-\epsilon H,
        \qquad
        \hat S_\epsilon(p)
        =
        \int_p^1\frac{dt}{\hat Q_\epsilon(t)},
    \]
and
    \[
        \hat T_\epsilon(p)
        :=
        \frac{\hat S_\epsilon(p)-S(p)}{\epsilon}
        =
        \int_p^1
        \frac{H(t)}{Q(t)\hat Q_\epsilon(t)}\,dt.
    \]
Then $\hat T_\epsilon(p)\to T_H(p)$ for every $p\in(0,1]$.

Since $\hat Q_\epsilon\leq Q$, the perturbed bidder is now the weaker bidder. Let the perturbed bidder be bidder 1, which is consistent with the labeling convention \eqref{relabel}. By symmetry of the two bidders,
    \[
        R_H(-\epsilon)
        =
        \mathrm{Rev}(Q,Q-\epsilon H)
        =
        \mathrm{Rev}(\hat Q_\epsilon,Q).
    \]
Let $\hat\eta_\epsilon$ be the inverse matching function from the stronger bidder's quantile to the weaker bidder's quantile. It satisfies
    \[
        \hat S_\epsilon(\hat\eta_\epsilon(p))=S(p).
    \]
By \eqref{eq:revenue},
    \[
        R_H(-\epsilon)
        =
        \int_0^1
        (2-p-\hat\eta_\epsilon(p))
        \hat Q_\epsilon(\hat\eta_\epsilon(p))\,dp.
    \]
Therefore,
    \[
        \frac{R_H(-\epsilon)-R_H(0)}{\epsilon}
        =
        \hat A_\epsilon-\hat B_\epsilon-\hat C_\epsilon,
    \]
where
    \[
        \hat A_\epsilon
        =
        \frac{2}{\epsilon}
        \int_0^1(1-p)
        \left[Q(\hat\eta_\epsilon(p))-Q(p)\right]dp,
    \]
    \[
        \hat B_\epsilon
        =
        \frac{1}{\epsilon}
        \int_0^1
        (\hat\eta_\epsilon(p)-p)Q(\hat\eta_\epsilon(p))\,dp,
    \]
and
    \[
        \hat C_\epsilon
        =
        \int_0^1
        (2-p-\hat\eta_\epsilon(p))
        H(\hat\eta_\epsilon(p))\,dp.
    \]

    \begin{claim}\label{claim:general-left-AB}
        \[
            \hat B_\epsilon
            \xrightarrow{\epsilon\to0^+}
            \int_0^1T_H(p)Q(p)^2\,dp,
        \]
and
        \[
            \hat A_\epsilon
            \xrightarrow{\epsilon\to0^+}
            2\int_{(0,1)}(1-p)T_H(p)Q(p)\,dQ(p).
        \]
    \end{claim}

    \begin{proof}
The proof follows the same structure as the proof of Claims~\ref{claim:general-right-B} and \ref{claim:general-right-A}. The only change is that the perturbed quantile is now the weaker quantile.

First,
        \[
            \hat S_\epsilon(p)
            -
            \hat S_\epsilon(\hat\eta_\epsilon(p))
            =
            \hat S_\epsilon(p)-S(p)
            =
            \epsilon\hat T_\epsilon(p).
        \]
Hence
        \[
            \int_p^{\hat\eta_\epsilon(p)}
            \frac{dt}{\hat Q_\epsilon(t)}
            =
            \epsilon\hat T_\epsilon(p).
        \]
Since $\hat Q_\epsilon$ is monotone,
        \[
            \frac{\hat\eta_\epsilon(p)-p}
            {\hat Q_\epsilon(\hat\eta_\epsilon(p))}
            \leq
            \int_p^{\hat\eta_\epsilon(p)}
            \frac{dt}{\hat Q_\epsilon(t)}
            \leq
            \frac{\hat\eta_\epsilon(p)-p}
            {\hat Q_\epsilon(p)}.
        \]
This implies
        \[
            \hat T_\epsilon(p)\hat Q_\epsilon(p)
            Q(\hat\eta_\epsilon(p))
            \leq
            \frac{(\hat\eta_\epsilon(p)-p)
            Q(\hat\eta_\epsilon(p))}{\epsilon}
            \leq
            \hat T_\epsilon(p)
            \hat Q_\epsilon(\hat\eta_\epsilon(p))
            Q(\hat\eta_\epsilon(p)).
        \]
Since $\hat\eta_\epsilon(p)\to p^+$, the middle term converges to $T_H(p)Q(p)^2$ for a.e. $p$. The dominated convergence argument is the same as before. Indeed, for $\epsilon$ small enough,
        \[
            \hat T_\epsilon(p)
            \leq
            \frac{C}{1-C\epsilon}S(p),
        \]
where $C$ is the constant such that $H\leq CQ$, and
        \[
            S(p)=\hat S_\epsilon(\hat\eta_\epsilon(p))
            \leq
            \frac{1}{1-C\epsilon}S(\hat\eta_\epsilon(p)).
        \]
Hence the upper bound is uniformly dominated by a constant multiple of
        \[
            S(\hat\eta_\epsilon(p))Q(\hat\eta_\epsilon(p))^2,
        \]
which is bounded by $Q(1)$. Therefore,
        \[
            \hat B_\epsilon
            \to
            \int_0^1T_H(p)Q(p)^2\,dp.
        \]

For $\hat A_\epsilon$, use the same Stieltjes measure $\mu=dQ$. Then
        \[
            Q(\hat\eta_\epsilon(p))-Q(p)
            =
            \mu([p,\hat\eta_\epsilon(p))).
        \]
By Fubini,
        \[
            \hat A_\epsilon
            =
            2\int_{(0,1)}
            \frac{1}{\epsilon}
            \left(
                \int_{\hat p_\epsilon(t)}^t(1-p)\,dp
            \right)dQ(t),
        \]
where
        \[
            \hat p_\epsilon(t)=
            \begin{cases}
                S^{-1}(\hat S_\epsilon(t))
                & \text{if } \hat S_\epsilon(t)<S(0),\\
                0
                & \text{if } \hat S_\epsilon(t)\geq S(0).
            \end{cases}
        \]
For every fixed $t\in(0,1)$, when $\epsilon$ is small enough, $\hat p_\epsilon(t)>0$ and
        \[
            S(\hat p_\epsilon(t))=\hat S_\epsilon(t).
        \]
Thus
        \[
            \int_{\hat p_\epsilon(t)}^t\frac{du}{Q(u)}
            =
            \hat S_\epsilon(t)-S(t)
            =
            \epsilon\hat T_\epsilon(t).
        \]
Since $\hat p_\epsilon(t)\to t^-$, the local average of $1/Q$ over $(\hat p_\epsilon(t),t)$ converges to $1/Q(t)$. Hence
        \[
            \frac{t-\hat p_\epsilon(t)}{\epsilon}
            \to
            T_H(t)Q(t).
        \]
Therefore,
        \[
            \frac{1}{\epsilon}
            \int_{\hat p_\epsilon(t)}^t(1-p)\,dp
            \to
            (1-t)T_H(t)Q(t).
        \]
The same bound as above gives domination:
        \[
            \frac{1}{\epsilon}
            \int_{\hat p_\epsilon(t)}^t(1-p)\,dp
            \leq
            \frac{t-\hat p_\epsilon(t)}{\epsilon}
            \leq
            \hat T_\epsilon(t)Q(t)
            \leq
            \frac{C}{1-C\epsilon}S(t)Q(t)
            \leq
            2C.
        \]
Dominated convergence with respect to the finite measure $dQ$ proves the claimed limit for $\hat A_\epsilon$.
    \end{proof}

    \begin{claim}\label{claim:general-left-C}
        \[
            \hat C_\epsilon
            \xrightarrow{\epsilon\to0^+}
            2I_H.
        \]
    \end{claim}

    \begin{proof}
Since for a fixed small $\epsilon_1>0$, $H=\frac{(Q+\epsilon_1 H)-Q}{\epsilon_1}$ is a difference of two bounded non-decreasing functions, $H$ has bounded variation and therefore has at most countably many discontinuities. Since $\hat\eta_\epsilon(p)\to p^+$, we have
        \[
            H(\hat\eta_\epsilon(p))\to H(p)
            \qquad \text{for a.e. }p.
        \]
Also,
        \[
            2-p-\hat\eta_\epsilon(p)\to 2(1-p).
        \]
The integrand is dominated by $2C Q(1)$, since $H\leq CQ$. Hence dominated convergence gives
        \[
            \hat C_\epsilon
            \to
            \int_0^1 2(1-p)H(p)\,dp
            =
            2I_H.
        \]
    \end{proof}

By Claim~\ref{claim:general-left-AB}, the limits of $\hat A_\epsilon$ and $\hat B_\epsilon$ combine exactly as in the right-derivative calculation:
    \[
        \lim_{\epsilon\to0^+}
        \left(\hat A_\epsilon-\hat B_\epsilon\right)
        =
        I_H-B_H-J_H.
    \]
The difference from the right derivative is the additional term $\hat C_\epsilon$, which appears because the perturbed bidder is now the weaker bidder and hence its quantile function enters directly in the revenue integrand. Therefore,
    \[
        \lim_{\epsilon\to0^+}
        \frac{R_H(-\epsilon)-R_H(0)}{\epsilon}
        =
        I_H-B_H-J_H-2I_H
        =
        -I_H-B_H-J_H.
    \]
Equivalently,
    \[
        (R_H)'_-(0)
        =
        \lim_{\epsilon\to0^+}
        \frac{R_H(-\epsilon)-R_H(0)}{-\epsilon}
        =
        I_H+B_H+J_H.
    \]

It remains only to verify the expression for $B_H$. Since $H\leq CQ$, we have
    \[
        T_H(p)\leq C S(p).
    \]
If $Q(0)=0$, then
    \[
        T_H(p)Q(p)^2
        \leq
        C S(p)Q(p)^2
        \leq
        C(1-p)Q(p)
        \to 0
        \qquad \text{as }p\to0^+.
    \]
Hence $B_H=0$. If $Q(0)>0$, then $T_H(0)<\infty$ and $T_H(p)\to T_H(0)$ as $p\to0^+$. Therefore,
    \[
        B_H=T_H(0)Q(0)^2.
    \]
This completes the proof.
\end{proof}

\subsection{Proof of Proposition~\ref{prop:global-weakening}}\label{Appendix:global weaken}
\begin{proof}
We label the bidder with quantile function $W$ as bidder 1 and the bidder with quantile function $Q$ as bidder 2, which is consistent with \eqref{relabel}.

Define
    \[
        S_Q(p)=\int_p^1\frac{dt}{Q(t)},
        \qquad
        S_W(p)=\int_p^1\frac{dt}{W(t)}.
    \]
Since $W\leq Q$, we have $S_W\geq S_Q$. Let $\eta$ be the inverse matching function from bidder 2's quantile to bidder 1's quantile, i.e. $\int_p^1\frac{dt}{Q(t)}=\int_{\eta(p)}^1\frac{dt}{W(t)}$. By \eqref{eq:revenue}, we have
    \begin{equation}\label{eq:revenue W Q}
        \mathrm{Rev}(W,Q)=
        \int_0^1 (2-p-\eta(p))W(\eta(p))\,dp.
    \end{equation}
At the symmetric benchmark,
    \begin{equation}\label{eq:revenue Q Q}
        \mathrm{Rev}(Q,Q)
        =
        2\int_0^1(1-p)Q(p)\,dp.
    \end{equation}
We next rewrite these expressions in survival quantiles. Making the change of variable $u=1-t$ in $\int_{\eta(p)}^1\frac{dt}{W(t)}=\int_p^1\frac{dt}{Q(t)}$ and then setting $r=1-p$, we obtain
    \[
    \int_0^{1-\eta(1-r)}\frac{du}{W(1-u)}=\int_0^r \frac{du}{Q(1-u)}
    \]
Define
    \begin{equation}\label{eq:hly}
         h(u)=Q(1-u),
        \quad
        \ell(u)=W(1-u),
        \quad
        y(r)=1-\eta(1-r).
    \end{equation}
Then $h$ and $\ell$ are weakly decreasing, and $\ell\leq h$. The matching equation becomes
    \[
        \int_0^{y(r)}\frac{du}{\ell(u)}
        =
        \int_0^r\frac{du}{h(u)}.
    \]
Since $\ell\leq h$, the left-hand integral grows weakly faster. Hence $y(r)\leq r$. Making the change of variable $p=1-r$ in \eqref{eq:revenue W Q} and \eqref{eq:revenue Q Q}, we obtain
    \begin{equation}\label{eq:rev WQ}
        \mathrm{Rev}(W,Q)
        =
        \int_0^1(r+y(r))\ell(y(r))\,dr,
    \end{equation}
and
    \begin{equation}\label{eq:rev QQ}
        \mathrm{Rev}(Q,Q)
        =
        2\int_0^1 r h(r)\,dr.
    \end{equation}
We now introduce a common time parametrization $t$, which allows us to divide the difference of the two expected revenues into two terms that are both nonnegative. Denote
    \[
        T:=\int_0^1\frac{du}{h(u)}=S_Q(0).
    \]
The value $T$ may be $+\infty$. For each $t<T$, define $r(t)$ and $y(t)$ by\footnote{In \eqref{eq:hly}, $y$ is a function of $r$. With a slight abuse of notation, we write $y(t)\coloneq y(r(t))$.}
    \begin{equation}\label{eq:time parametrize}
         \int_0^{r(t)}\frac{du}{h(u)}=t,
        \qquad
        \int_0^{y(t)}\frac{du}{\ell(u)}=t.
    \end{equation}
These functions are well-defined. Indeed, the first integral is continuous and strictly increasing in $r$ on $[0,1)$, with range $[0,T)$. The second integral has range containing $[0,T)$, because $\ell\leq h$. If $T<\infty$, define $r(T)=1$ and $y(T)=\lim_{t\uparrow T}y(t)$. Since $\ell\leq h$, we have $y(t)\leq r(t)$ for every $t<T$. Since both $h$ and $\ell$ are bounded from above, $1/h$ and $1/\ell$ are bounded from below. Hence $r(t)$ and $y(t)$ are absolutely continuous. Differentiating \eqref{eq:time parametrize} with respect to $t$, we obtain
    \[
    \frac{r'(t)}{h(r(t))}=1,\quad \frac{y'(t)}{\ell(y(t))}=1
    \]
for a.e. $t<T$. Define $a(t)=h(r(t)),b(t)=\ell(y(t))$. Then, for a.e. $t<T$, $r'(t)=a(t), y'(t)=b(t)$.

Through a change of variable $r=r(t)$, \eqref{eq:rev WQ} becomes
    \[
        \mathrm{Rev}(W,Q)
        =
        \int_0^T (r(t)+y(t))a(t)b(t)\,dt.
    \]
Similarly
    \[
        \mathrm{Rev}(Q,Q)
        =
        2\int_0^T r(t)a(t)^2\,dt.
    \]
If $T=\infty$, these integrals are understood as improper integrals over $[0,\tau]$, followed by $\tau\uparrow\infty$. Since bidders' values are bounded, expected revenue is bounded, hence the improper integrals always converge to finite numbers.\footnote{Because $y(t)\leq r(t)$ and $h$ is weakly decreasing, $h\geq \ell$ does not imply $a\geq b$: it is possible that $a(t)<b(t)$ for some $t$. Further calculation is therefore needed to compare $\mathrm{Rev}(Q,Q)$ and $\mathrm{Rev}(W,Q)$.}

We prove the comparison on every compact interval $[0,\tau]$ with $\tau<T$, and then let $\tau\uparrow T$.\footnote{If $T<\infty$, $\tau$ in the following steps can be directly replaced by $T$.} Define
    \[
    \begin{aligned}
        D(\tau)
        &\coloneq 
        2\int_0^{\tau}r(t) a(t)^2\,dt
        -
        \int_0^{\tau}(r(t)+y(t))ab\,dt  \\
        &=
        D_1(\tau)+D_2(\tau),
    \end{aligned}
    \]
where
    \[
        D_1(\tau)
        =
        \int_0^{\tau}r(t) a(t)^2\,dt
        -
        \int_0^{\tau}y(t) b(t)^2\,dt,
    \]
and
    \[
        D_2(\tau)
        =
        \int_0^{\tau}
        \left[
            r(t) a(t)^2+y(t) b(t)^2-(r(t)+y(t))a(t)b(t)
        \right]dt.
    \]

    \begin{claim}\label{claim:global-D1}
For every $\tau<T$,
        \[
            D_1(\tau)\geq0.
        \]
    \end{claim}

    \begin{proof}
Since $r'(t)=a(t)$ for a.e.\ $t<T$, making the change of variable $u=r(t)$, we have
        \[
            \int_0^{\tau} r(t)a(t)^2\,dt
            =
            \int_0^{r(\tau)} u h(u)\,du.
        \]
Similarly, since $y'(t)=b(t)$ for a.e.\ $t<T$, we have
        \[
            \int_0^{\tau} y(t)b(t)^2\,dt
            =
            \int_0^{y(\tau)} u \ell(u)\,du.
        \]
Therefore,
        \[
        \begin{aligned}
            D_1(\tau)
            &=
            \int_0^{r(\tau)}u h(u)\,du
            -
            \int_0^{y(\tau)}u \ell(u)\,du  \\
            &=
            \int_0^{y(\tau)}u[h(u)-\ell(u)]\,du
            +
            \int_{y(\tau)}^{r(\tau)}u h(u)\,du
            \geq0,
        \end{aligned}
        \]
where the last inequality follows from $h\ge \ell$ and $r\ge y$.
    \end{proof}

    \begin{claim}\label{claim:global-D2}
For every $\tau<T$,
        \[
            D_2(\tau)\geq0.
        \]
    \end{claim}

    \begin{proof}
Let $z(t)=r(t)-y(t)$. Then $z$ is absolutely continuous on $[0,\tau]$, and
    \[
        z(t)\geq 0,
        \qquad
        z'(t)=a(t)-b(t)
    \]
for a.e. $t\in [0,\tau]$. Suppressing the common argument $t$, we have
        \[
        \begin{aligned}
            r a^2+y b^2-(r+y)ab
            &=
            (a-b)(ra-yb)  \\
            &=
            z'\cdot (za+yz').
        \end{aligned}
        \]
Hence
        \[
            D_2(\tau)
            =
            \int_0^{\tau}a z z'\,dt
            +
            \int_0^{\tau}y\cdot (z')^2\,dt.
        \]
The second term is nonnegative.

It remains to show that the first term is nonnegative. The function $a(t)=h(r(t))$ is weakly decreasing because $h$ is weakly decreasing and $r(t)$ is weakly increasing. Hence the function $a$ has bounded variation on $[0,\tau]$. The function $z$ is absolutely continuous on $[0,\tau]$. Thus $z^2$ is absolutely continuous. Through integration by parts for Lebesgue-Stieltjes integrals,
        \[
        \begin{aligned}
            \int_0^{\tau} a(t)z(t)z'(t)\,dt
            &=
            \frac12\int_0^{\tau}a(t)\,d(z(t)^2)  \\
            &=
            \frac12 a(\tau)z(\tau)^2
            -
            \frac12 a(0)z(0)^2
            -
            \frac12\int_{[0,\tau]}z(t)^2\,da(t).
        \end{aligned}
        \]
Since $z(0)=0$, the second term is zero. Since $a$ is weakly decreasing, the signed Stieltjes measure $da$ is nonpositive. Hence
        \[
            -\frac12\int_{[0,\tau]}z(t)^2\,da(t)\geq0.
        \]
The first term is also nonnegative. Therefore,
        \[
            \int_0^{\tau}a(t)z(t)z'(t)\,dt\geq0.
        \]
It follows that $D_2(\tau)\geq0$.
    \end{proof}

By Claims~\ref{claim:global-D1} and \ref{claim:global-D2},
    \[
        D(\tau)\geq0
    \]
for every $\tau<T$. Letting $\tau\uparrow T$ gives
    \[
        \mathrm{Rev}(Q,Q)-\mathrm{Rev}(W,Q)\geq0.
    \]
Finally suppose that $W<Q$ on a set of positive Lebesgue measure. Then $\ell<h$ on a set of positive Lebesgue measure. Letting $\tau\uparrow T$ in the proof of Claim~\ref{claim:global-D1} yields
    \[
    \begin{aligned}
        \lim_{\tau\uparrow T}D_1(\tau)
        &=
        \int_0^{\bar y}u[h(u)-\ell(u)]\,du
        +
        \int_{\bar y}^1u h(u)\,du,
    \end{aligned}
    \]
where
    \[
        \bar y=\lim_{t\uparrow T}y(t).
    \]
If $\ell<h$ on a set of positive measure in $[0,\bar y]$, then the first term is strictly positive. Or else, then $\ell<h$ on a set of positive measure in $(\bar y,1]$. In that case $\bar y<1$, and the second term is strictly positive. Hence
    \[
        \lim_{\tau\uparrow T}D_1(\tau)>0.
    \]
Since $D_2(\tau)\geq0$ for every $\tau<T$, we conclude that
    \begin{equation*}
        \mathrm{Rev}(Q,Q)>\mathrm{Rev}(W,Q).
        \qedhere
    \end{equation*}
\end{proof}

\subsection{Proof of Example~\ref{ex:weakening-weaker-raises-revenue}}\label{appendix:example}
Since
\[
    Q_2(p)\equiv 1,
\]
and
\[
    Q_x(p)=
    \begin{cases}
        1/6, & 0\le p\le 1/2\\
        x, & 1/2<p\le 1,
    \end{cases}
\]
we have $Q_x\le Q_2$ for every $x\in(1/2,1]$. Since $Q_2\equiv 1$,
\[
    S_2(p)=\int_p^1 \frac{dt}{Q_2(t)}=1-p.
\]
Moreover,
\[
    S_x(p)=\int_p^1 \frac{dt}{Q_x(t)}
    =
    \begin{cases}
        3-6p+\dfrac{1}{2x}, & 0\le p\le 1/2,\\[6pt]
        \dfrac{1-p}{x}, & 1/2<p\le 1.
    \end{cases}
\]
The cutoff $\alpha_x$ is determined by $S_x(\alpha_x)=S_2(0)=1$. For $x\in(1/2,1]$, this gives
\[
    \alpha_x=\frac13+\frac{1}{12x}\in[5/12,1/2).
\]
Since $S_2^{-1}(s)=1-s$, the matching function is
\[
    \phi_x(p)=S_2^{-1}(S_x(p))=1-S_x(p),
\]
that is,
\[
    \phi_x(p)=
    \begin{cases}
        6p-2-\dfrac{1}{2x}, & \alpha_x\le p\le 1/2,\\[6pt]
        1-\dfrac1x+\dfrac{p}{x}, & 1/2<p\le 1.
    \end{cases}
\]
Using the revenue formula \eqref{eq:revenue1},
\[
    \mathrm{Rev}(Q_x,Q_2)
    =
    \int_{\alpha_x}^1
    (2-p-\phi_x(p))Q_2(\phi_x(p))\,dp.
\]
Because $Q_2\equiv 1$ and $\phi_x(p)=1-S_x(p)$, this becomes
\[
\begin{aligned}
    \mathrm{Rev}(Q_x,Q_2)
    &=
    \int_{\alpha_x}^1 (1-p+S_x(p))\,dp\\
    &=
    \int_{\alpha_x}^{1/2}
    \left(4-7p+\frac{1}{2x}\right)dp
    +
    \int_{1/2}^1
    (1-p)\left(1+\frac1x\right)dp.
\end{aligned}
\]
Substituting $\alpha_x=1/3+1/(12x)$ and simplifying yields
\[
    R(x):=\mathrm{Rev}(Q_x,Q_2)
    =
    \frac{11}{36}+\frac{5}{72x}-\frac{5}{288x^2}.
\]
Hence
\[
    R'(x)
    =
    -\frac{5(2x-1)}{144x^3}.
\]
Thus $R'(x)<0$ for every $x>1/2$. Therefore, for any $1/2< \tilde x<x\le 1$, we have $Q_{\tilde x}\le Q_x\le Q_2$, but
\[
    \mathrm{Rev}(Q_{\tilde x},Q_2)>\mathrm{Rev}(Q_x,Q_2).
\]
That is, weakening the already weaker bidder can strictly raise expected revenue.

\end{document}